\documentclass[12pt]{article}

\usepackage[margin=1in]{geometry}
\usepackage{times}
\usepackage{graphicx}
\usepackage{booktabs}
\usepackage{amsmath}
\usepackage{amssymb}
\usepackage{amsthm}
\usepackage{natbib}
\usepackage{setspace}
\usepackage{caption}
\usepackage{subcaption}
\usepackage{array}
\usepackage{xcolor}
\usepackage{hyperref}

\hypersetup{
    colorlinks,
    linkcolor={blue!60!black},
    citecolor={blue!60!black},
    urlcolor={blue!60!black}
}

\newtheorem{proposition}{Proposition}
\newtheorem{theorem}{Theorem}

\begin{document}
\doublespacing
\setlength{\emergencystretch}{2em}
\newcommand{\AltMarginAggregateWidthMax}{0.0212}
\newcommand{\AltMarginMedianWidthMax}{0.0037}
\newcommand{\AltMarginMedianWidthMin}{0.0031}
\newcommand{\BindingCellEndpointCount}{42}
\newcommand{\BindingTopCellMaxShadow}{0.0691}
\newcommand{\BindingTopCellOD}{LA-LA}
\newcommand{\BindingTopCellSector}{Manufacturing}
\newcommand{\BindingTopCellTargetPct}{1.7}
\newcommand{\CIFrontierDistanceTV}{0.578}
\newcommand{\CIHomeFrontierMean}{0.478}
\newcommand{\CfsGoodsLogisticsWidthSharePct}{38.7}
\newcommand{\DualMaxGap}{0}
\newcommand{\DualMaxMomentShadow}{0.0180}
\newcommand{\FireWidthSharePct}{26.5}
\newcommand{\GravityFrontierDistanceTV}{0.123}
\newcommand{\GravityHomeFrontierMean}{0.122}
\newcommand{\GulfArkansasCIRank}{36}
\newcommand{\GulfArkansasGravityRank}{3}
\newcommand{\GulfCIPointsBelowBounds}{1}
\newcommand{\GulfCaliforniaCIRank}{3}
\newcommand{\GulfCaliforniaGravityRank}{10}
\newcommand{\GulfGravityPointsOutsideBounds}{0}
\newcommand{\GulfLossRatio}{0.860}
\newcommand{\GulfLouisianaCIPoint}{7.5e-05}
\newcommand{\GulfLouisianaCIRank}{1}
\newcommand{\GulfLouisianaGravityPoint}{0.0024}
\newcommand{\GulfLouisianaGravityRank}{1}
\newcommand{\GulfLouisianaLower}{0.0009}
\newcommand{\GulfLouisianaUpper}{0.0043}
\newcommand{\GulfMississippiCIRank}{2}
\newcommand{\GulfMississippiGravityRank}{2}
\newcommand{\GulfNewYorkCIRank}{5}
\newcommand{\GulfNewYorkGravityRank}{20}
\newcommand{\GulfRegionalTV}{0.137}
\newcommand{\GulfTopOverlap}{0.598}
\newcommand{\HighIndustryDomarError}{0}
\newcommand{\HighIndustryRegionalTV}{0}
\newcommand{\HighRegionSectorRankCorr}{0.569}
\newcommand{\HighRegionSectorTV}{0.419}
\newcommand{\HighRegionSectorTopOverlap}{0.334}
\newcommand{\HighRegionalTV}{0.352}
\newcommand{\HighRegionalTopOverlap}{0.495}
\newcommand{\HomeBiasCIMean}{0.028}
\newcommand{\HomeBiasGravityMean}{0.384}
\newcommand{\HomeBiasObservedMean}{0.506}
\newcommand{\InfoBandedAggregateLower}{0.0066}
\newcommand{\InfoBandedAggregateUpper}{0.0085}
\newcommand{\InfoBandedAggregateWidth}{0.0019}
\newcommand{\InfoBandedMedianWidth}{0.0036}
\newcommand{\InfoBilateralMedianWidthChange}{0.0004}
\newcommand{\InfoBilateralMedianWidthReductionPct}{11.9}
\newcommand{\InfoFinalAggregateInterval}{[0.0066, 0.0066]}
\newcommand{\InfoFinalAggregateLower}{0.0066}
\newcommand{\InfoFinalAggregateUpper}{0.0066}
\newcommand{\InfoFinalAggregateWidth}{0}
\newcommand{\InfoFinalMedianWidth}{0.0027}
\newcommand{\InfoFinalPninetyWidth}{0.0034}
\newcommand{\InfoFinalPossibleTop}{51}
\newcommand{\InfoFinalRobustTop}{0}
\newcommand{\InfoMarginsMedianWidth}{0.0047}
\newcommand{\InfoPairwiseDeterminedPct}{0.6}
\newcommand{\InfoPairwiseDeterminedShare}{0.006}
\newcommand{\InfoPairwiseTopTwentyPct}{1.0}
\newcommand{\InfoPairwiseTopTwentyShare}{0.010}
\newcommand{\LPValidationBandViolations}{0}
\newcommand{\LPValidationMaxColGap}{1.0e-07}
\newcommand{\LPValidationMaxRowGap}{1.0e-07}
\newcommand{\LPValidationMaxSupportLeak}{0}
\newcommand{\LocalMeanIntervalWidth}{0.714}
\newcommand{\LocalSupportWidthSharePct}{7.8}
\newcommand{\ManufacturingCIHomeShare}{0.026}
\newcommand{\ManufacturingObservedHomeShare}{0.371}
\newcommand{\ManufacturingSplitBandedAggregateWidth}{0.0023}
\newcommand{\ManufacturingSplitBandedMedianWidth}{0.0054}
\newcommand{\ManufacturingSplitDetailedCfsCoveragePct}{88.5}
\newcommand{\ManufacturingSplitManufacturingWidthSharePct}{41.2}
\newcommand{\ManufacturingSplitMaxMarginBandViolation}{0}
\newcommand{\ManufacturingSplitMaxSupportLeak}{0}
\newcommand{\ManufacturingSplitMaxTopFifteenCfsPct}{44.7}
\newcommand{\ManufacturingSplitMedianWidth}{0.0043}
\newcommand{\ManufacturingSplitMinCfsPairs}{1577}
\newcommand{\ManufacturingSplitPairwiseDeterminedPct}{0.0}
\newcommand{\ManufacturingSplitPninetyWidth}{0.0059}
\newcommand{\ManufacturingSplitTopWidthGroup}{services and other pooled sectors}
\newcommand{\ManufacturingSplitTopWidthGroupPct}{49.5}
\newcommand{\ManufacturingWidthSharePct}{21.1}
\newcommand{\MaxEmpiricalSpectralRadius}{0.4483}
\newcommand{\MaxMomentDistanceTau}{0.141}
\newcommand{\MaxMomentHomeTau}{0.141}
\newcommand{\MeanMomentDistanceTau}{0.035}
\newcommand{\MeanMomentHomeTau}{0.035}
\newcommand{\MiningExcludedMedianWidth}{0.0031}
\newcommand{\MiningExcludedPninetyWidth}{0.0042}
\newcommand{\MomentBandFormalConfidenceCount}{0}
\newcommand{\MomentBandRobustMedianWidth}{0.0031}
\newcommand{\MomentBandRobustPninetyWidth}{0.0042}
\newcommand{\NonlinearContainsPointCompletions}{1}
\newcommand{\NonlinearOuterWidth}{36.079}
\newcommand{\NonlinearPerturbationRadius}{2.947}
\newcommand{\NonlinearRemainderBound}{17.958}
\newcommand{\PbsWidthSharePct}{19.4}
\newcommand{\ServicesOtherWidthSharePct}{53.5}
\newcommand{\SharpGulfMedianWidth}{0.0027}
\newcommand{\SharpLocalPninetyWidth}{0.0001}
\newcommand{\SharpShipmentMedianWidth}{0.0006}
\newcommand{\WholesaleCIHomeShare}{0.028}
\newcommand{\WholesaleObservedHomeShare}{0.606}
\newcommand{\WidthTopGroup}{services and other pooled sectors}
\newcommand{\WidthTopGroupPct}{53.5}
\newcommand{\WidthTopGroupShare}{0.535}

\renewcommand{\thefootnote}{\fnsymbol{footnote}}

\begin{center}
{\Large \textbf{Partial Identification of Spatial Production Networks}\textsuperscript{*}\par}
\vspace{1.5em}

\begin{tabular}{>{\centering\arraybackslash}m{0.3\textwidth} >{\centering\arraybackslash}m{0.3\textwidth} >{\centering\arraybackslash}m{0.3\textwidth}}
\textbf{Shaowen Luo} & \textbf{Kwok Ping Tsang} & \textbf{Zichao Yang} \\
\end{tabular}

\vspace{0.75em}
{\today\par}
\end{center}

\footnotetext[1]{Luo: Department of Economics, Virginia Tech, \href{mailto:sluo@vt.edu}{sluo@vt.edu}. Tsang: Department of Economics, Virginia Tech, \href{mailto:byront@vt.edu}{byront@vt.edu}. Yang: Wenlan School of Business, Zhongnan University of Economics and Law, \href{mailto:yang\_zichao@outlook.com}{yang\_zichao@outlook.com}.}
\setcounter{footnote}{0}
\renewcommand{\thefootnote}{\arabic{footnote}}

\begingroup
\setstretch{1.35}
\begin{abstract}
Which regional exposure conclusions are identified when public data do not observe buyer-seller links across states? We study this question by treating the missing intermediate-input spatial kernel as an unknown coupling constrained by regional activity margins, support restrictions, and auxiliary shipment moments. For linear exposure statistics, the sharp identified set is computed by transportation linear programs. Applying the method to U.S. state-sector data, we find that shipment data are inconsistent with the spatial diffuseness implied by proportional regionalization in key goods sectors. However, they do not identify a unique regional production network or a precise ranking of state exposure to local shocks. Bilateral shipment restrictions tighten the bounds, but much of the remaining uncertainty comes from large service and mixed sectors that are weakly covered by goods-movement data. The results show which exposure conclusions are supported by public data and which are imposed by maintained regionalization assumptions.
\end{abstract}
\endgroup



\clearpage

\section{Introduction}

Regional production-network calculations require assumptions about who buys from whom across space. A national input-output table tells us how much each industry buys from each other industry. Regional activity data tell us where industries are located. They do not tell us which supplier states sell intermediate inputs to which buyer states. The missing matrix is a bilateral state-sector buyer-seller matrix.

This paper studies what can be learned about that matrix before imposing a regionalization rule. We focus on the intermediate-input spatial kernel, the joint distribution of supplier and buyer locations for purchases from each supplier industry. Public data restrict its origin and destination margins only after the researcher chooses proxy mappings from regional activity to intermediate supply and demand. They do not identify the cells of the joint distribution. A proportional allocation, gravity completion, or support rule is therefore a maintained restriction on an unidentified coupling, not an observed regional IO table.

We develop a partial-identification approach to this problem. The admissible set consists of all nonnegative spatial kernels that satisfy maintained proxy margins, support restrictions, and auxiliary shipment moments. For linear one-step exposure statistics, the sharp identified interval is computed by transportation linear programs. This lets us separate exposure claims that survive the admissible set from claims that are imposed by a particular completed regional network.

The distinction matters because the same missing links are harmless for some propagation questions and first order for others. A pure national industry shock has no within-industry geographic variation, so the spatial kernel cancels in first-round exposure and in the Leontief accounting multiplier. This does not imply that industry shocks have small aggregate effects. It means that they cannot identify the spatial kernel. Local shocks are different. For regional and region-sector shocks, the kernel determines where downstream exposure lands, even when aggregate exposure looks stable.

We make the point in a regional trade-production model. The national input-output table gives the industry input share. What it does not give is the spatial sourcing share that says which supplier locations sell to which buyer locations. In an unrestricted regionalization, the regional intermediate-spending coefficient is
\begin{equation}
\label{eq:main_object_intro}
W_{(r,i),(s,j)}=\omega_{ji}\pi^{ij}_{r|s}.
\end{equation}
The national IO table identifies \(\omega_{ji}\), but not \(\pi^{ij}_{r|s}\). The empirical application imposes a lower-dimensional supplier-sector kernel, so buyer industries share the same spatial sourcing pattern within a supplier sector. This reduces the dimensionality of the missing bilateral matrix, but it does not make the remaining supplier-sector coupling identified.

The empirical application uses U.S. states, sixteen sectors, national IO coefficients, QCEW wage-bill and employment proxies for state-sector activity, and the 2017 Commodity Flow Survey. The CFS moments show that the proportional conditional-independence completion is too spatially diffuse in shipment-covered sectors. The average within-state shipment share is \HomeBiasObservedMean{} in the CFS, compared with \HomeBiasCIMean{} under conditional independence and \HomeBiasGravityMean{} under a structured gravity completion. Conditional independence also has a mean distance-bin total variation gap of \CIFrontierDistanceTV, compared with \GravityFrontierDistanceTV{} for structured gravity.

These shipment moments restrict the admissible set, but they do not identify local-shock incidence. For a Gulf regional shock, conditional independence lies outside the final sharp exposure interval for the directly shocked state of Louisiana: its point exposure is \GulfLouisianaCIPoint{}, below the final sharp lower endpoint of \GulfLouisianaLower{}. Thus the proportional completion is not simply one admissible point inside the final set. The same calculation also shows the limit of the available public data. Adding bilateral CFS cell bands lowers the median state exposure interval width by \InfoBilateralMedianWidthReductionPct{} percent relative to home-share and distance-bin moments. The final exact-margin specification still determines only \InfoPairwiseDeterminedPct{} percent of sharp state-pair exposure rankings, and only \InfoPairwiseTopTwentyPct{} percent of pairs involving the 20 states with the largest upper endpoints. No state is guaranteed to be in the top exposure decile, and all \InfoFinalPossibleTop{} states remain possible top-decile states.

The residual interval width is not concentrated only in the shipment-covered goods sectors. Services and other pooled sectors account for \ServicesOtherWidthSharePct{} percent of total state interval width, with FIRE and professional and business services alone accounting for \FireWidthSharePct{} and \PbsWidthSharePct{} percent. CFS-covered goods and logistics account for \CfsGoodsLogisticsWidthSharePct{} percent. Additional goods-shipment information therefore narrows the set, but it cannot by itself resolve state-level incidence when large service-sector sourcing relationships remain weakly observed. When wage-bill and employment margins are treated as bands rather than exact margins, aggregate exposure is no longer pinned down by exact-margin aggregation. Its width is \InfoBandedAggregateWidth{}, while the median state width remains \InfoBandedMedianWidth{}.

We make three contributions. First, we characterize sharp bounds for unknown regional production-network couplings. The endpoints are transportation linear programs, and the dual variables summarize which margins or moments bind a given bound. Second, we compare common regionalization restrictions using CFS shipment moments and public state-sector activity proxies. The comparison covers conditional independence, structured gravity completion, local support restrictions, bilateral CFS cell bands, and banded margins. Third, we apply the admissible set to local-shock exposure and show which state-incidence conclusions survive. The analysis does not estimate a full regional IO table, employment effect, output effect, or welfare effect. A full counterfactual still requires final-demand sourcing, behavioral parameters, and equilibrium closure. A causal event study also needs a design for the outcome response.

\subsection{Relation to the literature}

The closest substantive literature is macroeconomic work on production networks and spatial propagation. Input-output linkages shape the transmission of shocks across sectors and, in regional models, across locations \citep{Hulten1978,LongPlosser1983,Horvath1998,Horvath2000,Gabaix2011,FoersterEtAl2011,AcemogluEtAl2012,Atalay2017,CaliendoEtAl2018,BaqaeeFarhi2019}. A natural structural reference point is \citet{CaliendoEtAl2018}, who specify a quantitative spatial model to study regional and sectoral productivity shocks, equilibrium reallocation, aggregate effects, and welfare. \citet{AdaoArkolakisEsposito2020} provide a complementary reduced-form bridge between shift-share designs and spatial general equilibrium effects by estimating bilateral reduced-form elasticities across local labor markets. We ask a different question. We study what public data identify about one input into these calculations before imposing the trade, substitution, labor-market, final-demand, and market-clearing structure needed for a full counterfactual.

The analysis is also related to evidence on firm-level production networks. Firm-to-firm data show that actual buyer-seller links matter for shock transmission \citep{BarrotSauvagnat2016,BoehmEtAl2019,CarvalhoEtAl2021}. Those papers use more detailed link-level information than is available in public regional IO construction. We ask what remains identifiable when the researcher has national industry linkages, state-sector proxy margins, and shipment moments, but not buyer-seller intermediate-input links.

The paper also uses partial identification and data combination. The identified-set logic follows the partial-identification perspective of \citet{Manski2003}, \citet{Tamer2010}, and \citet{Molinari2020}. The empirical problem also resembles data-combination settings in which separate sources identify different margins or moments of an unobserved joint distribution \citep{CrossManski2002,RidderMoffitt2007}. We do not report confidence intervals for the identified set. If sampling inference over estimated bounds were added, the relevant econometric issues would be those studied by \citet{ImbensManski2004} and \citet{Stoye2009}.

Finally, the computation uses the same mathematical structure as discrete optimal transport with fixed marginals. We use that structure as an identification and accounting tool, not as a behavioral transport-cost model. The linear-program representation is standard in optimal transport \citep{Galichon2016}. The regional IO literature has long developed non-survey, partial-survey, location-quotient, and balancing methods for constructing regional and interregional tables \citep{MillerBlair2009,FleggEtAl1995,FleggWebber2000,BoeroEtAl2018}. We do not propose a new regionalization formula. We provide an identified-set analysis of the intermediate-input spatial kernel that such calculations often take as an input.

The remainder of the paper proceeds as follows. Section \ref{sec:general_framework} defines the spatial-kernel coupling problem, the maintained proxy margins, and the sharp linear-programming bounds. Section \ref{sec:shock_irf} explains which exposure and multiplier calculations depend on the kernel before structural closure. Section \ref{sec:ci} describes conditional independence, structured gravity completion, support restrictions, and CFS moment bands. Section \ref{sec:implementation} presents the U.S. state-sector evidence. Section \ref{sec:experiments} reports the identified-set exposure results. Section \ref{sec:empirical_applications} compares the Gulf regional shock with a national manufacturing shock. Section \ref{sec:conclusion} concludes. The appendix gives proofs, data details, nonlinear multiplier calculations, and robustness tables.

\section{Spatial Kernels and Sharp Bounds}
\label{sec:general_framework}
\label{sec:model}
\label{sec:object}

The empirical problem is a missing joint distribution. A national input-output
table gives supplier-industry spending shares \(\omega_{ji}\) and buyer
industry intermediate-input intensities \(\mu_j\). Their product is
\begin{equation}
\label{eq:Bji}
B_{ji}=\mu_j\omega_{ji}.
\end{equation}
The missing component is the spatial coupling for intermediate purchases of a
fixed supplier sector \(i\). Let \(K^i_{rs}\) be the share of supplier-\(i\)
intermediate purchases that pairs supplier state \(r\) with buyer state \(s\).
The destination-conditional sourcing share is
\begin{equation}
\label{eq:model_joint}
\pi^i_{r|s}=\frac{K^i_{rs}}{b^i_s},
\end{equation}
where \(b^i_s\) is the destination margin. The regional input coefficient is
\begin{equation}
\label{eq:model_A}
A^K_{(r,i),(s,j)}=B_{ji}\pi^i_{r|s}.
\end{equation}
Thus the national IO block identifies industry linkages, while \(K^i\)
allocates those linkages across space.

The baseline uses maintained proxy margins. The origin margin is the
state-sector activity share,
\begin{equation}
\label{eq:intermediate_origin_marginal}
a^i_r=\frac{x_{ri}}{\sum_\ell x_{\ell i}},
\end{equation}
with \(x_{ri}\) measured by the relevant state-sector activity proxy. The
destination margin measures where supplier-\(i\) inputs are used:
\begin{equation}
\label{eq:intermediate_destination_marginal}
D^I_{si}=\sum_j \mu_j\omega_{ji}x_{sj},
\qquad
b^i_s=\frac{D^I_{si}}{\sum_\ell D^I_{\ell i}}.
\end{equation}
These margins are not observed bilateral intermediate-input flows. They are
maintained mappings from public state-sector activity data and the national IO
block.

\label{subsec:two_layers}
This creates two layers of non-identification. First, state-sector activity
does not separately identify output sold to intermediate users, household
final demand, and residual final use. Second, even if the relevant origin and
destination totals were known, the bilateral matrix matching supplier states
to buyer states would remain unidentified. The baseline therefore studies the
narrower intermediate-input kernel. Final-demand sourcing and equilibrium
closure are left to Appendix~\ref{app:final_demand_accounting}.

\begin{proposition}[Non-identification from national IO and regional marginals]
\label{prop:nonidentification}
Fix the national IO matrix. In the unrestricted pair-specific case, suppose
that for an industry pair the researcher observes origin and destination
margins that have the same total mass. If there are at least two regions, then
the joint spatial coupling is not identified by the national IO matrix and
those margins. Under the supplier-sector restriction used below, even if the
researcher observes the corresponding supplier-sector margins, the joint
spatial coupling is still not identified. Unless the relevant margins are
degenerate, there are multiple couplings that match the same observed margins
but imply different regional networks.
\end{proposition}

A two-region example makes the non-identification concrete. Let the origin
and destination margins both be \((1/2,1/2)\). The two couplings
\[
\begin{pmatrix}
1/2 & 0\\
0 & 1/2
\end{pmatrix}
\qquad \text{and} \qquad
\begin{pmatrix}
0 & 1/2\\
1/2 & 0
\end{pmatrix}
\]
match the same margins. The first is entirely local sourcing, while the
second is entirely cross-region sourcing. National IO coefficients and
regional margins alone cannot distinguish them.

For a generic finite coupling problem, let \(K\) be a nonnegative matrix with
row margin \(a\) and column margin \(b\). Support restrictions set selected
entries to zero. Auxiliary moments have the form
\[
m_h(K)=\sum_{x,y}g_{h,xy}K_{xy}.
\]
With a support set \(\mathcal S\) and a moment set \(\mathcal M\), the
admissible set is
\begin{equation}
\label{eq:generic_admissible_set}
\mathcal A(a,b,\mathcal S,\mathcal M)
=
\{K\geq0:\sum_y K_{xy}=a_x,\ \sum_x K_{xy}=b_y,\ 
K_{xy}=0 \text{ outside }\mathcal S,\ m(K)\in\mathcal M\}.
\end{equation}
Conditional independence is one feasible point when \(K_{xy}=a_xb_y\).
Structured gravity is another point completion. Support restrictions and CFS
moment bands define subsets of the transport polytope.

\begin{theorem}[Sharp linear identified sets and LP dual]
\label{thm:generic_lp_dual}
Suppose \(\mathcal M=\{m:\underline m\leq m\leq \overline m\}\), the
admissible set in equation \eqref{eq:generic_admissible_set} is nonempty, and
\(T(K)=c'k\) is linear in the vectorized coupling \(k\). Then the sharp
identified set for \(T(K)\) is \([\underline T,\overline T]\), where
\begin{equation}
\label{eq:generic_primal_min}
\underline T
=
\min_{k\geq0} c'k
\quad
\text{s.t.}
\quad
P_X k=a,\ P_Y k=b,\ Gk\leq\overline m,\ -Gk\leq-\underline m.
\end{equation}
Variables are restricted to \(\mathcal S\), and \(\overline T\) is obtained by
reversing the objective. The dual of the lower-bound program is
\begin{equation}
\label{eq:generic_dual}
\max_{u,v,\lambda^+,\lambda^-}
a'u+b'v+\overline m'\lambda^+-\underline m'\lambda^-,
\end{equation}
subject to
\[
u_x+v_y+\sum_{h=1}^J(\lambda^+_h-\lambda^-_h)g_{h,xy}
\leq c_{xy}
\quad \text{for all }(x,y)\in\mathcal S,
\]
with \(\lambda^+\leq0\) and \(\lambda^-\leq0\). Strong duality holds under the
maintained feasibility and boundedness conditions.
\end{theorem}

\begin{proposition}[Nested information]
\label{prop:nested_information}
Let \(\mathcal M_1\subseteq\mathcal M_0\). If the corresponding admissible
sets are nonempty, then the identified intervals for any target \(T\) satisfy
\[
\underline T(\mathcal M_0)\leq \underline T(\mathcal M_1)
\leq \overline T(\mathcal M_1)\leq \overline T(\mathcal M_0).
\]
\end{proposition}

When auxiliary moments are estimated or reconciled from imperfect data, we use
moment bands:
\begin{equation}
\label{eq:moment_band_set}
\mathcal A_\alpha
=
\{K\in\mathcal A(a,b,\mathcal S,\mathbb R^J):
\hat m_h-c_{h,\alpha}\leq m_h(K)\leq \hat m_h+c_{h,\alpha}
\text{ for all }h\}.
\end{equation}

\begin{proposition}[Moment-band coverage]
\label{prop:moment_band_coverage}
Let \(K_0\) be the true coupling and suppose it satisfies the maintained
margins and support restriction. If
\[
\Pr\left(
\hat m_h-c_{h,\alpha}\leq m_h(K_0)\leq \hat m_h+c_{h,\alpha}
\text{ for all }h
\right)\geq 1-\alpha,
\]
then the interval obtained by minimizing and maximizing a linear functional
\(T(K)\) over \(\mathcal A_\alpha\) covers \(T(K_0)\) with probability at
least \(1-\alpha\).
\end{proposition}

For exposure, the target is linear. Let \(Q(E)=\sum_{s,j}q_{sj}E_{sj}\). For a
shock \(z\),
\begin{equation}
\label{eq:linear_exposure_functional}
Q(E^K(z))
=
\sum_i\sum_{r,s}\ell^i_{rs}(q,z)K^i_{rs},
\qquad
\ell^i_{rs}(q,z)
=
\frac{z_{ri}}{b^i_s}\sum_j q_{sj}B_{ji}.
\end{equation}

\begin{proposition}[Sharp bounds for linear exposure]
\label{prop:admissible_exposure}
If each sectoral admissible set is nonempty, compact, convex, and defined by
linear margins, support restrictions, or moment bands, then the admissible
values of \(Q(E^K(z))\) form the interval obtained by minimizing and
maximizing equation~\eqref{eq:linear_exposure_functional} over the product of
sectoral admissible sets. The endpoints are sharp and are computed by
sector-by-sector transportation linear programs.
\end{proposition}

The theorem and propositions are the main tools used below. They also define
the limits of the analysis. The bounds are sharp for linear exposure given the
maintained margins and moment bands. They do not identify behavioral
elasticities, final-demand substitution, factor adjustment, or welfare.

\section{Propagation Before Structural Closure}
\label{sec:shock_irf}

What does the missing kernel affect? Given an admissible intermediate-input
kernel \(K\), define the regional input matrix
\begin{equation}
\label{eq:macro_A}
A^K_{(r,i),(s,j)}
=
B_{ji}\pi^i_{r|s}.
\end{equation}
The first-round exposure of destination node \((s,j)\) to a shock vector \(z\)
is
\begin{equation}
\label{eq:general_exposure}
E^K_{sj}(z)
=
\sum_i B_{ji}\sum_r \pi^i_{r|s}z_{ri}.
\end{equation}
This one-step exposure is linear in \(K\). With a unit shock, a value of
0.001 means one-tenth of one percentage point in this accounting exposure
index before equilibrium responses. The corresponding Leontief accounting
multiplier is
\begin{equation}
\label{eq:macro_multiplier}
M^K
=
\left(I-\left(A^K\right)'\right)^{-1},
\end{equation}
whenever \(\rho(A^K)<1\). A Domar-style accounting exposure index is
\begin{equation}
\label{eq:domar_loss}
\mathcal L^K(z)
=
\sum_{s,j}\lambda_{sj}\left[M^Kz\right]_{sj}.
\end{equation}
These are accounting quantities. Employment, output, and welfare responses
require final demand, prices, factor adjustment, financing, and market
clearing.

For a pure industry shock, \(z_{ri}=z_i\) for every origin \(r\), the
first-round exposure of buyer node \((s,j)\) is
\begin{equation}
\label{eq:industry_shock_exposure}
E_{sj}^{I}(z)
=
\sum_i B_{ji}z_i,
\end{equation}
because the sourcing shares sum to one. The same cancellation holds for
Leontief accounting exposure because each term in the Leontief series maps a
vector that is constant across origins within an industry into another vector
with the same property. Pure industry shocks therefore cannot validate a
spatial regionalization rule.

For a pure regional shock, \(z_{ri}=z_r\), the kernel generally determines
where exposure lands. One aggregate incidence measure still cancels. If
exposure is first aggregated over destination regions using the
supplier-sector destination margins \(b^i_s\), then
\begin{equation}
\label{eq:regional_aggregate_industry}
\bar E_j^R(z)
=
\sum_i B_{ji}\sum_s b^i_s\sum_r \pi^i_{r|s}z_r
=
\sum_i B_{ji}\sum_r a^i_r z_r .
\end{equation}

\begin{proposition}[Invariance and kernel dependence]
\label{prop:shock_cancellation}
For first-round exposure in equation~\eqref{eq:general_exposure}, every
\(K\in\mathcal A(\mathcal M)\) gives the same buyer-industry exposure to a
pure industry shock. In the Leontief accounting system in
equation~\eqref{eq:macro_multiplier}, the same invariance holds for the full
multiplier response to pure industry shocks. For a pure regional shock, the
kernel is not needed for the destination-margin-weighted industrial incidence
in equation~\eqref{eq:regional_aggregate_industry}. The kernel is needed to
allocate exposure across destination regions whenever the shock has geographic
content.
\end{proposition}

The proposition gives the paper's boundary result. Industry shocks can answer
industry propagation questions, but they say little about regional incidence.
Regional and region-sector shocks require the spatial kernel for local
incidence.

Leontief exposure is harder to bound sharply because \(M^K\) is nonlinear in
\(K\). Exact global bounds are nonlinear optimization problems. We therefore
use one local outer-bound calculation only as an appendix result. Appendix
\ref{app:additional_evidence} gives the sensitivity formula and the
conservative remainder bound used in the nonlinear multiplier table. The main
empirical results below are sharp for linear one-step exposure, not for the
full nonlinear multiplier.

\section{Admissible Spatial-Kernel Restrictions}
\label{sec:ci}

Common regionalization rules are restrictions on the feasible coupling set.
Some select a point completion. Others define an admissible subset over which
linear exposure can be bounded. The question is whether the restriction is
admissible relative to maintained margins, support restrictions, and auxiliary
shipment moments.

The proportional completion fills in the missing joint coupling as
\begin{equation}
\label{eq:ci}
K^{i,CI}_{rs}=a^i_r b^i_s .
\end{equation}
In destination-conditional form, every buyer region sources supplier-\(i\)
inputs from the same origin distribution. Conditional independence is also
the maximum-entropy completion, maximizing
\begin{equation}
\label{eq:entropy}
-\sum_{r,s}K^i_{rs}\log K^i_{rs}.
\end{equation}

\begin{proposition}[Conditional independence as maximum entropy]
\label{prop:max_entropy}
Among all nonnegative couplings with origin marginal \(a^i\) and destination
marginal \(b^i\), the proportional matrix \(K^{i,CI}_{rs}=a^i_r b^i_s\) is the
unique maximum-entropy coupling when the marginals are positive. Equivalently,
it imposes zero mutual information between supplier and buyer locations
conditional on supplier industry.
\end{proposition}

This makes conditional independence a useful benchmark, not an identified
coupling. It rules out home bias, distance decay, corridor structure, and other
dependence between supplier and buyer locations after conditioning on supplier
industry.

\label{subsec:dgp_ci_bias}
We also use a structured gravity completion as a point comparison. For each
supplier industry \(i\), the unbalanced kernel is
\begin{equation}
\label{eq:dgp_kernel}
\widetilde K^i_{rs}(\eta,\tau)
=
a^i_r b^i_s
\exp\left\{
\eta \mathbf{1}\{r=s\}
-\tau \log(1+d_{rs})
\right\}.
\end{equation}
It is rebalanced by iterative proportional fitting to match the maintained
margins.

\label{sec:workflow}
The empirical analysis uses different restrictions by sector. For
shipment-covered sectors, we estimate a gravity-style point completion,
\begin{equation}
\label{eq:gravity}
K^i_{rs}
\propto
\exp\left\{
\alpha^i_r+\delta^i_s+\eta_i\mathbf{1}\{r=s\}
-\tau_i \log(1+d_{rs})
\right\},
\end{equation}
where origin and destination fixed effects absorb the marginals. For mixed
sectors, we pool spatial parameters. For local sectors, we impose support
restrictions rather than treating shipment-like observations as sectoral
intermediate-input flows. The baseline local support includes same-state
pairs, adjacent-state pairs, and the minimum additional nearest-state radius
needed for feasibility.

Let \(\mathcal S_i\) be a support set and let
\(\mathcal K_i(a^i,b^i,\mathcal S_i)\) be the nonnegative matrices that match
the margins and place zero mass outside \(\mathcal S_i\). For any linear
exposure functional \(L(K^i)=\sum_{r,s}\ell_{rs}K^i_{rs}\), sharp bounds are
\begin{equation}
\label{eq:transport_bounds}
\underline L_i
=
\min_{K^i\in\mathcal K_i(a^i,b^i,\mathcal S_i)}
L(K^i),
\qquad
\overline L_i
=
\max_{K^i\in\mathcal K_i(a^i,b^i,\mathcal S_i)}
L(K^i).
\end{equation}
These are sharp under the maintained support restriction because the feasible
set is exactly the transport polytope matching \(a^i\), \(b^i\), and
\(\mathcal S_i\).

All restrictions are fixed before propagation outcomes are inspected. This
differs from a standard regional IO construction because a completed
matrix can fit shipment geography and still reveal little about
local-shock incidence.

\section{U.S. State-Sector Evidence}
\label{sec:implementation}

We next compare common spatial-kernel restrictions with observed shipment
geography. The empirical application uses U.S. states and a sixteen-sector
aggregation. The national input-output block provides the industry shares
\(\omega_{ji}\) and the intermediate-input intensities \(\mu_j\). The baseline
origin margins, destination margins, and exposure weights use 2019 QCEW
wage-bill shares. The destination marginal is constructed from destination
activity, sector intermediate-input intensities, and the national IO matrix,
as in equation~\eqref{eq:intermediate_destination_marginal}. It is a proxy
for intermediate-demand geography, not regional final expenditure. Wage bills
are the baseline because they are a public, consistently available measure of
state-sector economic activity and they weight high-productivity state-sector
cells more than headcount alone. QCEW employment shares are the first
robustness margin, and model-output margins are reported only in
robustness. State-to-state shipment flows come from the
2017 Commodity Flow Survey. State distances are computed from Census state
centroids.

We therefore interpret the empirical inputs as maintained proxy measures.
The spatial kernel is the unknown intermediate-input coupling.
Conditional independence is the proportional benchmark. Structured gravity
completion is a low-dimensional point completion estimated from shipment
geography. One-step exposure is an accounting exposure measure before
employment responses, price responses, and welfare effects.

\begin{table}[htbp]
\centering
\caption{Sector Treatment in the Main Specification}
\label{tab:sector_classification}
\normalsize
\setlength{\tabcolsep}{4pt}
\begin{tabular}{@{}p{0.22\textwidth}p{0.38\textwidth}p{0.32\textwidth}@{}}
\toprule
Group & Sectors & Treatment \\
\midrule
Shipment-covered & Manufacturing, mining, transportation, wholesale & CFS moment and selected bilateral cell bands \\
Tradable pooled & Agriculture & Pooled tradable parameters \\
Mixed pooled & Utilities, information, finance, insurance, real estate, professional and business services & Pooled mixed-sector parameters \\
Local support-restricted & Construction, retail, education, health care, leisure, other services, government & Support bounds for local sourcing \\
\bottomrule
\end{tabular}
\begin{flushleft}
\footnotesize Notes: The classification is fixed before propagation outcomes are interpreted. Retail is treated as local in the main specification because retail output is conceptually closer to local service provision than to an interregional supplier sector. Appendix Table~\ref{tab:retail_robustness} reports a shipment-informed retail variant.
\end{flushleft}
\end{table}

\subsection{Why shipment data do not observe the coupling}

CFS moments are auxiliary shipment moments, not the full intermediate-input
coupling. The distinction matters for interpretation. First, CFS
shipments include final goods as well as goods that may become intermediate
inputs. Second, CFS commodity classifications do not map one-for-one into
the production industries in the IO table. Third, wholesale shipments and
re-shipments can break the link between the producing origin and the
intermediate-input seller relevant for a buyer. Fourth, services are missing
or weakly covered, which matters because many local and business-service
sectors are large in the IO block. Finally, CFS observes goods movement. It
does not observe the buyer-seller use of intermediate inputs by destination
industry. For this reason, the CFS restrictions below restrict shipment
geography within selected sectors, but they do not convert the spatial kernel
into an observed matrix.

We use point completions only as comparisons. Conditional
independence is the proportional benchmark. Structured gravity, with pooled
variants where data are thin, is the shipment-informed comparison. For local
sectors, the relevant calculation is not a point completion but the support-restricted
admissible set. Lower and upper support-bound kernels summarize feasible
ranges. All comparisons hold fixed the same national IO block and the same
intermediate-demand marginals. They differ only in the intermediate-input
spatial coupling.

Figure~\ref{fig:home_bias_wedge} shows the main empirical fact. In the
four strict shipment-covered sectors, observed state-to-state CFS flows display
large within-state shares. The average observed CFS home share is
\HomeBiasObservedMean. Conditional independence implies \HomeBiasCIMean. The
structured gravity completion implies \HomeBiasGravityMean. Thus conditional
independence does not merely smooth bilateral flows. It removes most of the
same-state mass observed in shipment data. Appendix
Table~\ref{tab:home_bias_wedge} reports the corresponding home-share and
held-out RMSE values.

\begin{figure}[htbp]
\centering
\includegraphics[width=0.82\textwidth]{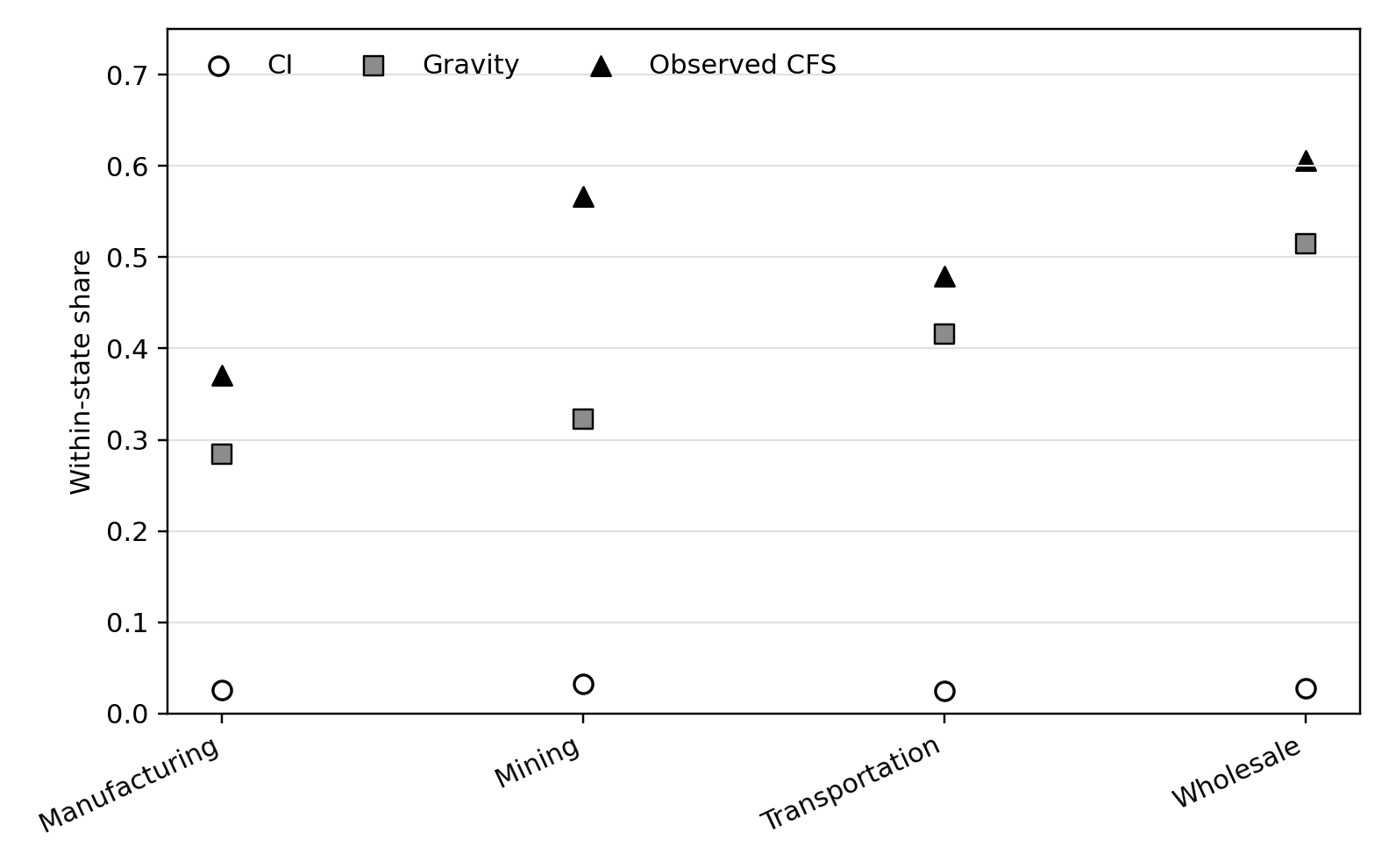}
\caption{Home Bias in Shipment-Covered Kernel Restrictions}
\label{fig:home_bias_wedge}
\begin{flushleft}
\footnotesize Notes: The figure reports within-state shipment or kernel shares for the four strict shipment-covered sectors. CFS is the observed within-state shipment share. CI is the proportional conditional-independence completion. Gravity is the structured spatial completion. The interpretation of CFS moments is described in Section~\ref{sec:implementation}.
\end{flushleft}
\end{figure}

Table~\ref{tab:admissibility_frontier} reports admissibility frontiers using
the same evidence. For each restriction, we ask how much the CFS moment restrictions must be relaxed before that restriction becomes admissible. A larger tolerance means the restriction is farther from the shipment evidence. Formally, for a restriction \(R\), let
\[
\tau_R=\inf\{\tau:m(K_R)\in\mathcal M(\tau)\}
\]
be the smallest tolerance under which the restriction is admissible for the maintained moment. For the shipment-covered sectors, the moments are the CFS home share, the CFS distance-bin distribution, and held-out flow fit. Conditional independence requires an average home-share tolerance of
\CIHomeFrontierMean, while sector gravity requires \GravityHomeFrontierMean.
The corresponding distance-bin total variation gaps are
\CIFrontierDistanceTV{} and \GravityFrontierDistanceTV. This should not be read as a formal statistical rejection of conditional independence, because it
does not use CFS sampling variances. It shows that conditional independence is admissible only under a much looser moment set than sector gravity.

\begin{table}[htbp]
\centering
\caption{Admissibility Frontier for Shipment-Covered Kernel Restrictions}
\label{tab:admissibility_frontier}
\normalsize
\begin{tabular}{lcccc}
\toprule
Restriction & Home gap & Distance TV & Avg. dist. gap & RMSE/width \\
\midrule
Conditional independence & 0.451 & 0.578 & 784.0 & 0.0120 \\
Sector gravity & 0.066 & 0.123 & 126.0 & 0.0061 \\
Support restriction & -- & -- & -- & 0.714 \\
\bottomrule
\end{tabular}
\begin{flushleft}
\footnotesize Notes: The table reports admissibility-frontier comparisons for the four strict shipment-covered sectors. Home gap is the mean absolute gap between
the completion-implied home share and the observed CFS home share. Distance TV is the mean total variation distance between the completion-implied and observed CFS distance-bin distributions. Avg. dist. gap is the mean absolute
gap in average shipment distance, in miles. RMSE is held-out normalized flow RMSE.
\end{flushleft}
\end{table}

Local support restrictions are different from the shipment-covered
point-completion comparisons in Table~\ref{tab:admissibility_frontier}. They define feasible sets rather than CFS flow-fit statistics. Appendix Figure~\ref{fig:local_bounds} reports the corresponding local-sector home-share intervals. The mean interval width is 0.714.

The gap is economically large. The proportional completion does more than smooth flows at the margin. It almost eliminates home bias in the sectors where shipment data show home bias most clearly. In manufacturing, observed CFS flows imply a within-state share of \ManufacturingObservedHomeShare{}, while the CI completion implies \ManufacturingCIHomeShare{}. In wholesale, the corresponding numbers are \WholesaleObservedHomeShare{} and
\WholesaleCIHomeShare{}. Gravity is closer on both home shares and held-out flow fit, but it remains a maintained low-dimensional restriction rather than an identified intermediate-input spatial kernel.

This distinction motivates the sector treatment in
Table~\ref{tab:sector_classification}: sector-specific CFS restrictions for
shipment-covered sectors, pooled parameters where shipment evidence is thin,
and support-restricted bounds for local sectors.

\section{Identified-Set Exposure Results}
\label{sec:experiments}

We now report sharp bounds for one-step exposure over the admissible set of
intermediate-input kernels. The bounds are not ranges across selected point
completions. They are the minimum and maximum exposure values attainable by
any kernel satisfying the maintained proxy margins, support restrictions, and
CFS moment bands. The calculations hold fixed the national IO table,
intermediate-demand proxy margins, sectoral intermediate-input intensities,
and state-sector wage-bill exposure weights. Only the intermediate-input
spatial coupling changes.

The results show that the public data determine some incidence
comparisons but leave many state rankings unresolved. The calibrated
economy has 51 states and 16 sectors. For each coupling \(K\), we construct
the regional input matrix in equation~\eqref{eq:macro_A} using
\(B_{ji}=\mu_j\omega_{ji}\), with
\(\mu_j=\text{intermediate}_j/\text{output}_j\). The maximum spectral radius
across all calibrated and comparator matrices is \MaxEmpiricalSpectralRadius{}, so the accounting
inverse exists in every reported case.

\subsection{Sharp one-step bounds for the Gulf regional shock}

We first study a Gulf regional shock that hits
Louisiana and Mississippi in all supplier sectors. This shock has geographic
content, so the spatial kernel matters for where exposure lands. At the same
time, the exact-margin aggregate target has a cancellation property. Because
the aggregate weights and destination margins use the same wage-bill proxy,
the unknown bilateral kernel collapses to maintained origin margins after
destination aggregation. The zero aggregate width in the exact-margin
baseline should therefore not be read as evidence that bilateral sourcing is
precisely identified. We therefore focus on state-level incidence.

Table~\ref{tab:information_content} reports how the Gulf exposure set changes
as information is added. The first four rows add exact proxy margins,
local-sector support restrictions, CFS home-share bands, and CFS distance-bin
bands. The fifth row adds selected bilateral CFS cell bands for shipment-
covered sectors. The final two rows replace exact wage-bill margins with
bands whose lower and upper endpoints are the QCEW wage-bill and employment
shares. The mean home-share feasibility tolerance is \MeanMomentHomeTau, and
the maximum is \MaxMomentHomeTau. Appendix Table~\ref{tab:moment_tolerances}
reports the sector-specific tolerances.

\begin{table}[htbp]
\centering
\caption{Information Content of Spatial Restrictions}
\label{tab:information_content}
\normalsize
\resizebox{\textwidth}{!}{\begin{tabular}{lccccc}
\toprule
Admissible set & Agg. interval & Agg. width & Median width & P90 width & Possible top \\
\midrule
Margins only & [0.0066, 0.0066] & 0 & 0.00474 & 0.00658 & 51 \\
Local support & [0.0066, 0.0066] & 0 & 0.00312 & 0.00421 & 51 \\
CFS home band & [0.0066, 0.0066] & 0 & 0.00309 & 0.00419 & 51 \\
CFS distance-bin band & [0.0066, 0.0066] & 0 & 0.00309 & 0.00419 & 51 \\
CFS bilateral cell bands & [0.0066, 0.0066] & 0 & 0.00272 & 0.00340 & 51 \\
\midrule
Banded margins & [0.0065, 0.0085] & 0.00197 & 0.00388 & 0.00553 & 51 \\
Banded margins + CFS bilateral & [0.0066, 0.0085] & 0.00189 & 0.00364 & 0.00470 & 51 \\
\bottomrule
\end{tabular}}
\begin{flushleft}
\footnotesize Notes: The table reports sharp one-step exposure bounds for the
Gulf regional shock. Agg. interval and Agg. width refer to the
wage-bill-weighted aggregate one-step exposure interval. Median width and P90
width summarize state-level exposure intervals. Possible top is the number of
states that can be in the top exposure decile for some admissible kernel,
using conservative interval classification. The banded-margin rows are a
separate admissible-set layer, not a nested refinement of the exact wage-bill
rows. No row has a robust top-decile state.
\end{flushleft}
\end{table}

The top-decile classification is deliberately conservative, so we also solve
sharp pairwise bounds. For each unordered state pair, we bound \(E_s-E_{s'}\).
One state is classified as dominating the other only when the sharp lower
bound is positive or the sharp upper bound is negative. Under the final
exact-margin bilateral-CFS set, only eight of 1,275 state pairs are
determined, or \InfoPairwiseDeterminedPct{} percent. Only one state has any
robust dominance relation. The determined share among pairs involving the 20
states with the largest upper endpoints is
\InfoPairwiseTopTwentyPct{} percent. This pattern is visible in Figure
\ref{fig:gulf_state_intervals}: except for Louisiana, most lower endpoints
among high-upper-bound states remain close to zero.

Figure~\ref{fig:gulf_state_intervals} shows why many rankings remain
unresolved. Most lower endpoints among the high-upper-bound states are close
to zero, and many intervals overlap even after CFS restrictions. The figure
also shows that conditional independence is outside the final sharp interval
for one directly shocked state. For Louisiana, conditional independence gives
exposure \GulfLouisianaCIPoint{}, below the sharp lower endpoint
\GulfLouisianaLower{}. Structured gravity remains inside all displayed
intervals.

\begin{figure}[htbp]
\centering
\includegraphics[width=0.86\textwidth]{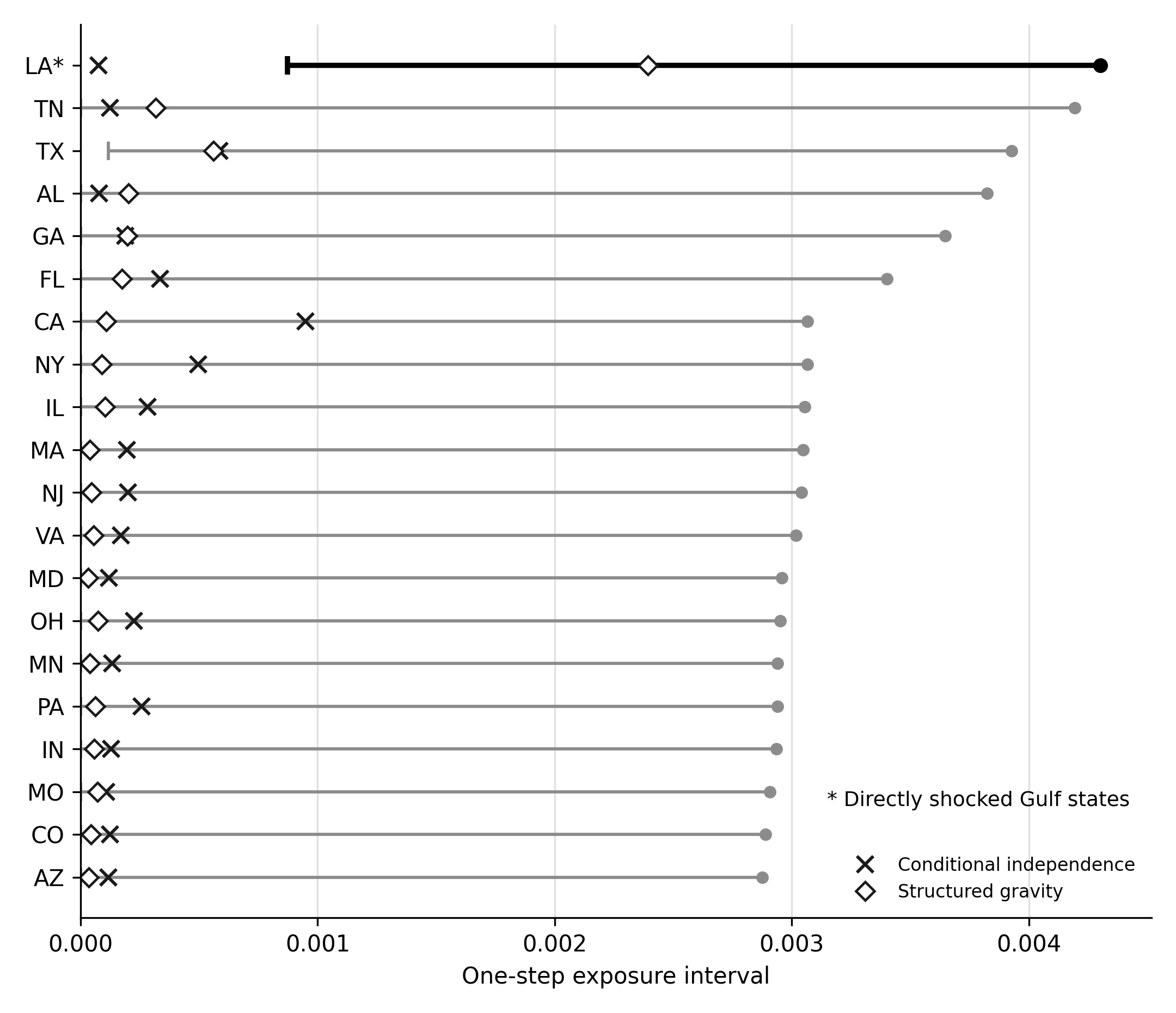}
\caption{Top Gulf State Exposure Intervals}
\label{fig:gulf_state_intervals}
\begin{flushleft}
\footnotesize Notes: Horizontal lines report the sharp one-step Gulf regional-
shock exposure interval for the 20 states with the largest upper endpoint.
States are sorted by the upper endpoint. Crosses mark conditional
independence, and diamonds mark structured gravity. Asterisks mark Louisiana
and Mississippi. Louisiana has a strictly positive lower endpoint, and its
conditional-independence point lies below that lower bound. Most other lower
endpoints are close to zero.
\end{flushleft}
\end{figure}

The sector decomposition in Table~\ref{tab:sector_width_decomposition}
explains where residual uncertainty comes from. The decomposition is exact:
the state exposure interval width is the sum of sector-level endpoint
differences because the transport problem separates by supplier sector.
FIRE, professional and business services, manufacturing, wholesale, and
transportation are the largest contributors to median state width. Grouping
the sectors shows that services and other pooled sectors account for
\ServicesOtherWidthSharePct{} percent of total state interval width. FIRE and
professional and business services alone account for \FireWidthSharePct{} and
\PbsWidthSharePct{} percent. CFS-covered goods and logistics account for
\CfsGoodsLogisticsWidthSharePct{} percent. The residual uncertainty is
therefore not only a problem of missing goods-shipment cells. It also reflects
large service and mixed sectors whose intermediate-input geography is not
directly observed in CFS.

\begin{table}[htbp]
\centering
\caption{Sector Decomposition of Gulf Interval Width}
\label{tab:sector_width_decomposition}
\normalsize
\resizebox{\textwidth}{!}{\begin{tabular}{llccc}
\toprule
Supplier sector & Group & Median contribution & P90 contribution & Mean share (\%) \\
\midrule
FIRE & services and other pooled sectors & 0.00076 & 0.00076 & 26.5 \\
Prof. Business Services & services and other pooled sectors & 0.00059 & 0.00059 & 19.4 \\
Manufacturing & CFS-covered goods and logistics & 0.00057 & 0.00057 & 21.1 \\
Wholesale & CFS-covered goods and logistics & 0.00028 & 0.00039 & 10.5 \\
Transportation & CFS-covered goods and logistics & 0.00014 & 0.00034 & 6.9 \\
Information & services and other pooled sectors & 0.00010 & 0.00010 & 4.0 \\
Utilities & services and other pooled sectors & 0.00008 & 0.00015 & 3.5 \\
Agriculture & services and other pooled sectors & 0.00004 & 0.00006 & 1.6 \\
\bottomrule
\end{tabular}}
\begin{flushleft}
\footnotesize Notes: Contributions are sector-level differences between the
upper and lower exposure endpoint for each state, aggregated over states. The
table lists the largest sector contributors by median state contribution.
\end{flushleft}
\end{table}

The bilateral CFS cells lower the median state width by
\InfoBilateralMedianWidthReductionPct{} percent, so they add information, but
the gain is modest relative to the remaining interval width. Appendix
Table~\ref{tab:binding_bilateral_cfs_cells} reports which selected bilateral
cells bind in the Gulf exposure endpoint problems.

The banded-margin rows in Table~\ref{tab:information_content} evaluate the
role of exact wage-bill margins. In those rows, each
sector's total mass remains one, but origin and destination shares are allowed
to range between the QCEW wage-bill and employment shares. Under banded
margins with bilateral CFS restrictions, aggregate exposure width is
\InfoBandedAggregateWidth{} and median state width is
\InfoBandedMedianWidth{}. These rows should be interpreted as a separate
admissible-set layer, not as a nested refinement of the exact wage-bill
baseline.

Appendix Table~\ref{tab:moment_band_bounds} reports the CFS moment-band
calculation. Appendix Table~\ref{tab:alternative_margin_identified_sets}
repeats the identified-set calculation with wage-bill, employment, mixed, and
model-output margins. The model-output row is a robustness specification rather than
the baseline. Across those margin constructions, all states remain possible
top-decile exposure states. Appendix Table~\ref{tab:mining_robustness} also
shows that excluding mining CFS moments leaves the median and p90 state
interval widths at \MiningExcludedMedianWidth{} and
\MiningExcludedPninetyWidth{}. The exact-margin aggregate cancellation is
therefore a property of the maintained aggregate target. Weak identification
of regional incidence is not.

Appendix Table~\ref{tab:dual_shadow_prices} reports the largest nonzero LP
moment shadows for state-level Gulf exposure bounds. The rows are not
whole-state exposure endpoints. They are sector-specific dual results inside
the state-level bounds. They show which CFS moment restrictions bind
particular sector-level endpoints, while the maximum primal-dual gap in the
full dual output is only \DualMaxGap{}.

Nonlinear Leontief multiplier exposure is harder to bound sharply. Appendix
Table~\ref{tab:nonlinear_multiplier_outer_bounds} reports the perturbation
outer-bound calculation. In the Gulf application, the first-order LP interval
is much narrower than the final outer interval because the conservative
perturbation radius, \NonlinearPerturbationRadius, is above one. The resulting
outer interval is valid and contains the reported point-completion multiplier
values, but it is too wide to support a sharp nonlinear identified-set claim
in this application. Without additional structure, nonlinear Leontief exposure
is much harder to sharply bound than linear one-step exposure. For this
reason, the empirical results in the main text are interpreted as sharp bounds
for one-step exposure, not for full nonlinear propagation.

\section{Local Shock Exposure Applications}
\label{sec:empirical_applications}

The preceding results imply a simple distinction. A shock can have economic
effects without identifying the spatial kernel. If the shock is
national within a supplier industry, every origin in that industry is hit and
the destination-conditional sourcing shares sum out. A shock with geographic
content is different because the spatial kernel determines where downstream
exposure lands.

Table~\ref{tab:empirical_application_summary} applies this distinction to two
shock designs. The first is a Katrina-style Gulf regional shock:
\[
z_{ri}^{Gulf}
=
\mathbf{1}\{r\in\{\mathrm{LA},\mathrm{MS}\}\}.
\]
The shock hits Louisiana and Mississippi in every supplier sector. Its
geographic content is sharp while its industry content is broad. The second is a national manufacturing-input shock:
\[
z_{ri}^{Mfg}
=
\mathbf{1}\{i=\mathrm{manufacturing}\}.
\]
It is best interpreted here as an industry-shock limiting case rather than a
detailed tariff counterfactual.

\begin{table}[htbp]
\centering
\caption{Gulf and Manufacturing Shock Comparisons}
\label{tab:empirical_application_summary}
\normalsize
\begin{tabular}{lrrrrr}
\toprule
Application & Domar error & Regional TV & Rank corr. & Top overlap & Loss ratio \\
\midrule
Gulf regional & 0.008 & 0.137 & 0.887 & 0.598 & 0.860 \\
Manufacturing industry & 0 & 0 & 1.000 & 1.000 & 1.000 \\
\bottomrule
\end{tabular}

\begin{flushleft}
\footnotesize Notes: The table compares conditional independence with the
structured-gravity completion. Domar error is the absolute difference in
output-weighted accounting loss. Regional TV compares state-level loss
allocations. Top overlap compares the top decile of state-sector losses.
\end{flushleft}
\end{table}

For the Gulf regional shock, the point-completion aggregate difference is not
negligible. The Domar-weighted loss under conditional independence is
\GulfLossRatio{} of the structured-gravity loss. The regional allocation also
changes substantially: state-level allocation TV is \GulfRegionalTV{}, and
top-decile overlap is \GulfTopOverlap{}. Thus the proportional completion
changes both aggregate accounting loss and which downstream places are
classified as highly exposed. For the national manufacturing shock, all
metrics are invariant up to numerical precision, as predicted by
Proposition~\ref{prop:shock_cancellation}.

The ranking changes are most visible outside the directly shocked states.
Louisiana and Mississippi remain the two most exposed destination states under
both completions. But Arkansas is ranked \GulfArkansasGravityRank{} under
structured gravity and \GulfArkansasCIRank{} under conditional independence.
California moves from rank \GulfCaliforniaGravityRank{} to
\GulfCaliforniaCIRank{}, and New York moves from rank
\GulfNewYorkGravityRank{} to \GulfNewYorkCIRank{}, under conditional
independence. These are point-completion comparisons, not sharp identified
rankings, but they show why regional incidence cannot be inferred from
aggregate exposure alone.

\section{Conclusion}
\label{sec:conclusion}

Regional production-network calculations require assumptions about who buys
from whom across space. Standard data provide national industry IO tables and
regional sectoral activity, but they do not observe the bilateral
state-sector buyer-seller matrix. This paper shows that those data identify
an admissible set of intermediate-input spatial kernels, not a unique regional
network or a full regional IO table. A completed regional IO matrix should
therefore be interpreted as a maintained restriction on the missing coupling,
not as an observed input.

This distinction matters because the missing spatial kernel is not equally
relevant for all propagation questions. Pure national industry shocks are
invariant to the kernel and therefore provide a limiting case for
regionalization assumptions. Local regional and region-sector shocks are
different because the kernel determines where downstream exposure lands. For
linear one-step exposure, the admissible set delivers sharp
transportation-program bounds. For nonlinear Leontief accounting exposure,
the same problem becomes a nonlinear identified-set problem. The perturbation
calculation provides conservative outer bounds rather than sharp global
bounds.

The U.S. state-sector application shows that the issue is quantitatively
relevant. CFS shipment moments imply substantially more home bias and distance
concentration than conditional independence, so the proportional completion is
too spatially diffuse in shipment-covered sectors. At the same time, these
moments do not identify the full intermediate-input coupling or the
state-level incidence of local shocks. For a Gulf regional shock, selected
bilateral CFS cells narrow the state exposure intervals, but they determine
few state-pair rankings. Banded wage-bill and employment margins remove the
exact-margin aggregate cancellation, yet state-level rankings remain weakly
identified.

These bounds can restrict quantitative spatial models, but they do not
replace them. A full counterfactual still requires final-demand sourcing,
household-demand geography, behavioral elasticities, factor adjustment,
financing, and market clearing. The point is to separate the
intermediate-input spatial network features supported by proxy margins and
auxiliary shipment moments from those imposed by regionalization assumptions.
Rather than evaluating a structural model at a single proportional
regionalization, researchers can ask whether its exposure and counterfactual
conclusions survive over the admissible set of spatial kernels.

\clearpage
\onehalfspacing
\bibliographystyle{aer}
\bibliography{state_spatial_kernel_refs}

\clearpage
\appendix

\section{Proofs}
\label{app:proofs}

\begin{proof}[Proof of Theorem \ref{thm:generic_lp_dual}]
The feasible set is a closed and bounded polytope because the variables are nonnegative and their total mass is fixed by the margins. A linear functional on this set attains its minimum and maximum, and its image is an interval. Sharpness follows because the admissible set contains exactly the couplings satisfying the maintained margins, support, and moment bands. The primal lower-bound problem is equation \eqref{eq:generic_primal_min}. The dual is the standard linear-programming dual for a minimization problem with equality constraints, nonnegative variables, and upper-bound inequalities. With the convention \(Gk\leq\overline m\) and \(-Gk\leq-\underline m\), the inequality multipliers are nonpositive, which gives equation \eqref{eq:generic_dual}. Feasibility and boundedness imply strong duality.
\end{proof}

\begin{proof}[Proof of Proposition \ref{prop:nested_information}]
If \(\mathcal M_1\subseteq\mathcal M_0\), then \(\mathcal A(\mathcal M_1)\subseteq\mathcal A(\mathcal M_0)\). Minimizing over a smaller feasible set cannot produce a lower value, and maximizing over a smaller feasible set cannot produce a higher value. This gives the stated inequalities.
\end{proof}

\begin{proof}[Proof of Proposition \ref{prop:moment_band_coverage}]
On the event that all true moments lie inside their bands, the true coupling \(K_0\) belongs to \(\mathcal A_\alpha\) because it also satisfies the maintained margins and support restrictions. The reported interval is the minimum and maximum of \(T(K)\) over a set that contains \(K_0\), so it contains \(T(K_0)\). The probability of that event is at least \(1-\alpha\).
\end{proof}

\begin{proof}[Proof of Proposition \ref{prop:nonidentification}]
Fix an industry pair \((i,j)\). The maintained origin and destination marginals impose \(2R-1\) independent linear constraints on the \(R^2\) entries of a nonnegative matrix \(K^{ij}\). If \(R\geq 2\), the number of unknown entries exceeds the number of independent equality constraints. When the marginals are nondegenerate, the feasible transport polytope has positive dimension. Therefore there is more than one matrix \(K^{ij}\) that matches the same marginals. Since \(W_{(r,i),(s,j)}=\omega_{ji}K^{ij}_{rs}/b^{ij}_s\) for destinations with positive \(b^{ij}_s\), each such matrix implies a different regional network while preserving the same national IO coefficients and regional marginals. The same argument applies to the restricted supplier-sector coupling \(K^i\) after replacing \((a^{ij},b^{ij})\) with \((a^i,b^i)\).
\end{proof}

\begin{proof}[Proof of Proposition \ref{prop:admissible_exposure}]
Substitute equation \eqref{eq:general_exposure} into \(Q(E)\):
\[
Q(E^K(z))
=
\sum_{s,j}q_{sj}\sum_i B_{ji}\sum_r \frac{K^i_{rs}}{b^i_s}z_{ri}
=
\sum_i\sum_{r,s}
\left(
\frac{z_{ri}}{b^i_s}\sum_j q_{sj}B_{ji}
\right)K^i_{rs}.
\]
Thus \(Q(E^K(z))\) is a linear functional of the entries of the sectoral kernels. The continuous image of a compact convex set under a linear functional is a compact convex subset of \(\mathbb R\), and hence an interval. The lower and upper endpoints are attained by Weierstrass' theorem. If the maintained restrictions are linear equalities, linear inequalities, or support restrictions, each \(\mathcal A_i(\mathcal M_i)\) is a polytope. Because the product set is separable across supplier sectors and the objective is additively separable, optimizing over \(\mathcal A(\mathcal M)\) is equivalent to optimizing the sector-specific linear objectives over the corresponding transportation polytopes and summing the endpoints. These bounds are sharp because every feasible point satisfies exactly the maintained restrictions and no other restrictions are imposed.
\end{proof}

\begin{proposition}[Multiplier sensitivity to kernel restrictions]
\label{prop:multiplier_sensitivity}
Suppose \(\rho(A^K)<1\). For a perturbation \(dK\) that induces \(dA^K\), the
first-order change in Leontief exposure is
\begin{equation}
\label{eq:multiplier_sensitivity}
d(M^Kz)
=
M^K(dA^K)'M^Kz.
\end{equation}
Thus, first-order multiplier sensitivity can be computed from the baseline
inverse and the matrix perturbation induced by the kernel perturbation.
\end{proposition}

\begin{proof}[Proof of Proposition \ref{prop:multiplier_sensitivity}]
Let \(M^K=(I-(A^K)')^{-1}\). The differential of a matrix inverse gives
\[
dM^K
=
M^K d(A^K)' M^K.
\]
Multiplying by the fixed shock vector \(z\) gives equation \eqref{eq:multiplier_sensitivity}. The formula is local because the inverse maps \(A^K\) into \(M^K\) nonlinearly.
\end{proof}

\begin{proposition}[Outer bounds for multiplier exposure]
\label{prop:certified_multiplier_bounds}
Fix a feasible baseline kernel \(K_0\), let \(A_0=A^{K_0}\), and let
\[
M_0=(I-A_0')^{-1}.
\]
For any admissible \(K\), write \(\Delta A=A^K-A_0\). Then
\begin{equation}
\label{eq:multiplier_expansion}
M^Kz
=
M_0z
+
M_0(\Delta A)'M_0z
+
\mathrm{Rem}_M(K;z),
\end{equation}
where
\begin{equation}
\label{eq:multiplier_remainder}
\mathrm{Rem}_M(K;z)
=
M_0(\Delta A)'M_0(\Delta A)'M^Kz .
\end{equation}
For any scalar weight vector \(q\), the first-order term
\[
q'M_0(\Delta A)'M_0z
\]
is linear in \(K\). Its extrema over \(\mathcal A(\mathcal M)\) are therefore
LP bounds. If
\[
\delta\geq\|(\Delta A)'\|_\infty,\qquad
r\geq\|(\Delta A)'M_0\|_\infty,\qquad
\overline M\geq\|M^K\|_\infty
\]
hold for all admissible \(K\), then
\begin{equation}
\label{eq:multiplier_remainder_bound}
|q'\mathrm{Rem}_M(K;z)|
\leq
\|q'M_0\|_1\, r\,\delta\,\overline M\,\|z\|_\infty .
\end{equation}
Adding this remainder radius to the LP bounds for the first-order term gives
an outer interval for \(q'M^Kz\). The interval is conservative and is not
claimed to be sharp.
\end{proposition}

\begin{proof}[Proof of Proposition \ref{prop:certified_multiplier_bounds}]
The resolvent identity gives
\[
M^K-M_0
=
M_0(A^K-A_0)'M^K.
\]
Multiplying by \(z\) and substituting \(M^K=M_0+(M^K-M_0)\) gives equation \eqref{eq:multiplier_expansion}. Since \(A^K\) is linear in \(K\), the scalar first-order term \(q'M_0(A^K-A_0)'M_0z\) is also linear in \(K\). Its extrema over the admissible set are therefore LP bounds by Theorem \ref{thm:generic_lp_dual}. The norm bound in equation \eqref{eq:multiplier_remainder_bound} follows from the duality between the \(1\)-norm and the infinity norm, together with submultiplicativity:
\[
|q'\mathrm{Rem}_M(K;z)|
\leq
\|q'M_0\|_1\|(\Delta A)'M_0(\Delta A)'M^Kz\|_\infty
\leq
\|q'M_0\|_1\,r\,\delta\,\overline M\,\|z\|_\infty .
\]
Adding and subtracting the resulting nonnegative radius from the first-order LP endpoints gives an interval that contains \(q'M^Kz\) for every admissible \(K\).
\end{proof}

\begin{proof}[Proof of Proposition \ref{prop:shock_cancellation}]
If \(z_{ri}=z_i\), then \(\sum_r \pi^i_{r|s}z_{ri}=z_i\sum_r\pi^i_{r|s}=z_i\), which gives equation \eqref{eq:industry_shock_exposure}. For the multiplier result, consider any vector \(x_{ri}=x_i\) that is constant across regions within each industry. Then
\[
\left[(A^K)'x\right]_{sj}
=
\sum_i B_{ji}\sum_r\pi^i_{r|s}x_i
=
\sum_i B_{ji}x_i,
\]
which is also constant across destination regions within buyer industry \(j\) and does not depend on \(K\). Therefore every term in the Leontief series applied to a pure industry shock is independent of the spatial kernel. If \(z_{ri}=z_r\), then aggregating the supplier-\(i\) exposure over destination regions with weights \(b^i_s\) gives \(\sum_s b^i_s\sum_r\pi^i_{r|s}z_r=\sum_{r,s}K^i_{rs}z_r=\sum_r a^i_rz_r\). Multiplying by \(B_{ji}\) and summing over supplier sectors gives equation \eqref{eq:regional_aggregate_industry}. Thus the kernel cancels only after the stated destination-margin aggregation. Without that aggregation, \(\sum_r\pi^i_{r|s}z_r\) varies with the destination-specific sourcing shares.
\end{proof}

\begin{proof}[Proof of Proposition \ref{prop:max_entropy}]
Maximize \(-\sum_{r,s}K_{rs}\log K_{rs}\) subject to \(\sum_s K_{rs}=a_r\), \(\sum_r K_{rs}=b_s\), and \(K_{rs}\geq 0\). With positive marginals, the entropy objective is strictly concave on the positive orthant. The first-order condition for an interior solution is \(-1-\log K_{rs}+\lambda_r+\mu_s=0\), so \(K_{rs}=A_rB_s\). The marginal restrictions imply \(K_{rs}=a_rb_s\). The mutual information of a coupling with these marginals is \(D_{KL}(K\|ab')\). It equals zero if and only if \(K=ab'\).
\end{proof}

\section{Data and Measures}
\label{app:data_construction}

The empirical application uses public inputs and reproducible tables and figures. The national IO block is the sixteen-sector production block used throughout the application. The baseline origin margins, destination margins, and exposure weights use 2019 QCEW wage-bill shares. The intermediate-demand destination margins are constructed using equation \eqref{eq:intermediate_destination_marginal}. QCEW employment shares provide the first robustness margin, and the model's generated steady-state output measure, \texttt{steady\_X\_is}, is retained only as a model-implied robustness margin. The state-to-state flow data come from the 2017 Commodity Flow Survey. State distances are computed from Census state centroids, and adjacency is derived from county adjacency files aggregated to the state level.

The empirical inputs are constructed in four steps. First, CFS shipment records are mapped to the sixteen-sector model classification, a balanced state-to-state flow grid is formed, origin and intermediate-demand destination proxy marginals are constructed, spatial parameters are estimated where possible, and baseline completions are produced. Second, the identified-set calculations construct admissible sets from exact margins, banded margins, local support restrictions, CFS home-share bands, CFS distance-bin bands, and selected bilateral CFS cell bands, then solve transportation linear programs for one-step exposure bounds, sector decompositions, moment-band intervals, dual shadow prices, and nonlinear multiplier outer bounds. Third, the exposure-accounting simulation combines the completed kernels with the national IO block, sectoral intermediate-input intensities, and state-sector exposure weights, then produces the point-completion multiplier tables and outputs for the Gulf regional and national manufacturing shock designs. Fourth, the paper tables and figures are constructed from the generated CSV outputs.

The strict sector classification is fixed before the results are interpreted. Manufacturing, mining, transportation, and wholesale are the shipment-covered sectors. Agriculture is pooled with tradable sectors. Utilities, information, finance-insurance-real estate, and professional and business services are pooled mixed sectors. Construction, retail, education, health care, leisure, other services, and government are local support-restricted sectors. Retail is also reported under a shipment-informed appendix classification.

The paper's figures, tables, and numerical macros are produced from CSV outputs rather than hand-entered values. This convention makes the numerical statements in the introduction and results sections traceable to source files. The home-bias means come from \path{home_bias_wedge_table.csv}, and the local-sector interval width comes from \path{local_identified_set_width.csv}. The identified-set outputs are \path{moment_band_bounds.csv}, \path{lp_dual_shadow_prices.csv}, and \path{nonlinear_multiplier_outer_bounds.csv}.

The sharp one-step exposure measures come from transportation linear programs that use the national production block, state-sector exposure weights, maintained proxy marginals, and admissible moment restrictions. The multiplier statistics come from the corresponding shock-level calculations, summary table, and matrix calculations.

The LP validation file reports a maximum row marginal error of \LPValidationMaxRowGap{} and a maximum column marginal error of \LPValidationMaxColGap{}. Maximum support leakage is \LPValidationMaxSupportLeak{}, and the number of moment-band violations is \LPValidationBandViolations{}. The dual table reports a maximum primal-dual gap of \DualMaxGap{}. These values show that the reported intervals are computed over feasible kernels that satisfy the maintained marginal, support, and shipment-moment restrictions up to numerical tolerance.

For the CFS-CV row in Table \ref{tab:moment_band_bounds}, the calculation uses the official coefficient of variation for value estimates in the downloadable 2017 CFS tables. For each sector and moment, it computes the numerator value, the denominator value, and conservative component standard errors from \texttt{VAL\_S}. The share half-width is
\[
z\left(\frac{\widehat{se}(N)}{D}+\frac{N\widehat{se}(D)}{D^2}\right),
\]
where \(N\) is the moment numerator and \(D\) is the sector total. The calculation uses \(z=3.13\), which is approximately a Bonferroni value for the 28 CFS moments used across four sectors. The construction is conservative. It ignores covariance terms because the public CFS files do not report the covariance matrix for the model-sector, state-pair, distance-bin shares.

\subsection{Measures}
\label{app:measure_definitions}

For each supplier sector \(i\) and completion \(m\), let \(K^{i,m}_{rs}\) denote the joint state-to-state coupling after balancing. The CFS measures below summarize observed goods shipments. They are auxiliary geography moments, not direct measures of the full intermediate-input coupling. The home share is
\begin{equation}
H_i^m=\sum_r K^{i,m}_{rr}.
\end{equation}
For observed CFS flows \(F^i_{rs}\), the observed home share is
\begin{equation}
H_i^{CFS} =
\frac{\sum_r F^i_{rr}}{\sum_{r,s}F^i_{rs}},
\end{equation}
when total observed flow is positive. The average distance of a completion is
\begin{equation}
D_i^m=\sum_{r,s}K^{i,m}_{rs}d_{rs}.
\end{equation}

Distance-bin shares are computed by summing observed flow or completed coupling mass over bins. For a bin \(B\), the CFS distance-bin share is
\begin{equation}
S_{iB}^{CFS}=
\frac{\sum_{(r,s):d_{rs}\in B}F^i_{rs}}{\sum_{r,s}F^i_{rs}},
\end{equation}
and the corresponding completed share is
\begin{equation}
S_{iB}^{m}=
\sum_{(r,s):d_{rs}\in B}K^{i,m}_{rs}.
\end{equation}
The bins used in the paper are same state, 0--250 miles, 250--500 miles, 500--1000 miles, 1000--1500 miles, and more than 1500 miles.

Held-out flow prediction is evaluated on a fixed deterministic set of origin-destination cells. Within each supplier sector, the observed held-out flows and predicted held-out kernel masses are normalized to sum to one. RMSE is then computed between the two normalized vectors. The scoring rule asks whether the completion predicts the distribution of held-out flows, not whether it predicts aggregate shipment volume.

For generic exposure disagreement, consider a one-state, one-supplier-sector shock to supplier sector \(i\) in origin state \(r\). For completion \(m\), the one-step exposure contribution to destination node \((s,j)\), scaled by the destination demand marginal for supplier \(i\), is
\begin{equation}
E_{s,j}^{m}(r,i)=B_{ji}K^{i,m}_{rs}
=\mu_j\omega_{ji}K^{i,m}_{rs}.
\end{equation}
These measures compare the vector \(E^m(r,i)\) across completions. The measure is an exposure-accounting statistic rather than the conditional sourcing share \(\pi^i_{r|s}\). Rank correlation compares the complete ordering of downstream state-sector exposures. Top-decile overlap compares the set of downstream nodes in the top ten percent of exposure. Maximum absolute disagreement records the largest difference in exposure across downstream nodes.

\section{From Spatial Kernels to Structural Multipliers}
\label{app:structural_closure}

The spatial kernel allocates intermediate-input exposure across supplier and buyer locations. It does not by itself determine the effect of a shock on aggregate output, state output, or welfare. Those quantities require additional closure assumptions. A full state-sector model also requires final-demand accounting assumptions and an equilibrium closure. The examples below show why financing and price responses are separate from the intermediate-input kernel.

A structural spatial model can be written as \((K,\theta,\mathcal C)\), where \(K\) is the spatial kernel, \(\theta\) collects behavioral elasticities and technology parameters, and \(\mathcal C\) is the equilibrium closure. The identified-set analysis restricts \(K\), but leaves \(\theta\) and \(\mathcal C\) to the structural model. A full spatial general equilibrium counterfactual can therefore be evaluated over admissible kernels \(K\in\mathcal A\), rather than over a single proportional regionalization.

\subsection{Final-demand accounting}
\label{app:final_demand_accounting}

The intermediate-input kernel and the final-demand sourcing kernel are distinct quantities. Let \(c_{si}\) denote domestic household final demand in destination state \(s\) for sector \(i\), after PCE categories have been concorded to production sectors, adjusted for imported content, and scaled to the domestic household final-use total for the sector. Let \(h^i_{r|s}\) denote the final-demand sourcing share from producer state \(r\) to household destination state \(s\). Household final demand for sector \(i\) output produced in \(r\) and absorbed by households in \(s\) is \(h^i_{r|s}c_{si}\), with \(\sum_r h^i_{r|s}=1\).

A transparent construction of this block requires assumptions beyond those used for the intermediate-input kernel. First, one must construct state-sector household expenditure shares from state PCE or impose national household expenditure shares where the concordance is weak. Second, imported content must be removed from household absorption, because imports are not domestic producer output. Third, exports, inventories, investment goods, federal purchases, and uncovered final uses should be kept in a residual final-use block unless their destination geography is observed. The residual block can be assigned to producer nodes, but assigning it to destination states is an additional spatial allocation rule.

Given a domestic household origin marginal \(a^{C,i}\) and destination marginal \(b^{C,i}\), the final-demand coupling can be written as
\begin{equation}
\label{eq:app_household_coupling}
K^{C,i}
=
\operatorname{IPF}\left(a^{C,i},b^{C,i};G^{C,i}\right),
\qquad
h^i_{r|s}
=
\frac{K^{C,i}_{rs}}{b^{C,i}_s},
\end{equation}
where \(G^{C,i}\) is a prior or support restriction for household final-demand sourcing. It need not equal the intermediate-input prior \(G^{I,i}\). Local services may require a local-support rule, while tradable goods may be better represented by shipment or gravity information. The same point applies over time. If both \(K^{I,i}\) and \(K^{C,i}\) are allowed to move freely during a shock period, changes in observed output or expenditure can always be rationalized by changes in the unobserved matrices. Empirical content requires fixed kernels or restricted, pre-specified time variation.

\subsection{A final-demand closure}
\label{app:simple_structural_closure}

The closure matters even in a deliberately simple case. Let \(f^R_n\) be a government final-demand impulse to producer node \(n\), let \(\hat F^R=\sum_{n\in\mathcal N}f^R_n\), and let \(\zeta^C_n\) denote the household final-demand basket. With lump-sum tax finance, the direct final-demand row perturbation is
\begin{equation}
\label{eq:app_tax_financed_impulse}
\delta f^R_n
=
f^R_n-\zeta^C_n\hat F^R.
\end{equation}
This perturbation sums to zero across producer nodes when \(\zeta^C\) sums to one. The zero-sum condition applies to nominal final demand. It does not imply that the real aggregate output response is zero, because the response also depends on relative prices and the production network. Let
\[
\hat P^C_w=\sum_{n\in\mathcal N}\zeta^C_n\hat p^w_n
\]
be the household consumption price-index response in wage-numeraire units. With exogenous labor, the aggregate real final-output response is
\begin{equation}
\label{eq:app_kappa_closure}
\hat Y^R
=
\hat F^R-\bar\Omega^C_0\hat P^C_w
=
\kappa^R\hat F^R,
\qquad
\kappa^R_{\mathrm{exog}\ L,\mathrm{tax}}
\equiv
1-\bar\Omega^C_0\frac{\hat P^C_w}{\hat F^R}.
\end{equation}
The shock can still change state-sector gross output through \(\delta f^R_n\) and its upstream propagation through the network. The parameter \(\kappa^R\) is not imposed by the budget constraint. It depends on the price response, and therefore on the production network and closure. A different financing rule, labor-market closure, or final-demand block would give a different mapping from exposure to output.

\subsection{Local comparative statics}
\label{app:local_ge_equations}

The same distinction can be written in the standard first-order production-network notation. Let \(\Omega\) be the matrix of initial expenditure shares for the final-demand node, producer nodes, and factor nodes, and let \(\Psi=(I-\Omega)^{-1}\) be the corresponding Leontief inverse. Let \(\lambda_n\) denote the Domar weight of producer node \(n\), and let hats denote first-order log changes around the initial equilibrium. A full local GE calculation would specify producer productivity shocks \(\hat A_n\), factor supply shocks \(\hat L_g\), final-demand composition shocks \(\hat\omega_{0n}\), residual final-demand level shocks, substitution elasticities, and factor-market closure.

With real prices measured relative to the GDP deflator, producer price changes satisfy
\begin{equation}
\label{eq:app_price_response}
\hat p_n
=
-
\sum_{k\in\mathcal N}\Psi_{nk}\hat A_k
+
\sum_{g\in\mathcal G}\Psi_{ng}
\left(\hat\lambda_g-\hat L_g+\hat Y\right).
\end{equation}
The first term is downstream cost propagation from productivity shocks. The second term is the cost effect of factor-price changes, which depends on the factor-market closure through the factor Domar weights \(\hat\lambda_g\).

Sales shares are determined by backward propagation through demand. For any producer or factor \(i\),
\begin{equation}
\label{eq:app_sales_share_response}
\lambda_i\hat\lambda_i
=
\theta_0
\operatorname{Cov}_{\Omega(0)}
\left(\hat\omega_0,\Psi_{\cdot i}\right)
+
\sum_{\ell\in\mathcal N}
\Delta^R_{0\ell}\Psi_{\ell i}
+
\sum_{j\in\{0\}\cup\mathcal N}
\lambda_j(\theta_j-1)
\operatorname{Cov}_{\Omega(j)}
\left(-\hat p,\Psi_{\cdot i}\right).
\end{equation}
The three terms capture final-demand composition shifts, residual final-demand row perturbations, and substitution responses to relative prices. The quantity response follows from the Domar identity:
\begin{equation}
\label{eq:app_node_quantity_response}
\hat y_n=\hat\lambda_n-\hat p_n+\hat Y.
\end{equation}
These equations show why exposure accounting is not an output or welfare counterfactual. The spatial kernel enters \(\Omega\) and \(\Psi\), but output responses also depend on the shocks, substitution elasticities, factor adjustment, final-demand block, and financing rule. The identified-set results in the main text restrict one input into this system, not the full counterfactual system.

\section{Nonlinear Multiplier Bound}
\label{app:additional_evidence}

Table \ref{tab:nonlinear_multiplier_outer_bounds} reports the nonlinear multiplier outer-bound calculation for the Gulf shock. The first-order LP interval is much narrower than the outer interval because the conservative perturbation radius exceeds one. The calculation is valid as an outer bound, but it is not informative enough to support a sharp nonlinear multiplier claim in this application.

\begin{table}[htbp]
\centering
\caption{Nonlinear Multiplier Outer Bounds}
\label{tab:nonlinear_multiplier_outer_bounds}
\normalsize
\setlength{\tabcolsep}{2pt}
\begin{tabular}{@{}>{\raggedright\arraybackslash}p{0.10\textwidth}>{\centering\arraybackslash}p{0.10\textwidth}>{\centering\arraybackslash}p{0.1475\textwidth}>{\centering\arraybackslash}p{0.075\textwidth}>{\centering\arraybackslash}p{0.08\textwidth}>{\centering\arraybackslash}p{0.115\textwidth}>{\centering\arraybackslash}p{0.1475\textwidth}>{\centering\arraybackslash}p{0.10\textwidth}@{}}
\toprule
Target & Baseline & \shortstack{First\\order} & Radius & \shortstack{Radius\\$<1$} & Remainder & \shortstack{Outer\\interval} & \shortstack{Contains\\points} \\
\midrule
Gulf Aggregate & 0.0553 & [-0.0018, 0.1602] & 2.947 & No & 17.958 & [-17.960, 18.119] & Yes \\
\bottomrule
\end{tabular}
\begin{flushleft}
\footnotesize Notes: The table reports the perturbation outer-bound calculation in Proposition \ref{prop:certified_multiplier_bounds}. First order is the LP interval for \(q'M_0(A^K-A_0)'M_0z\) added to the baseline multiplier value. Radius is the conservative bound on \(\|(A^K-A_0)'M_0\|_\infty\). The outer interval adds the computed remainder bound. Contains points indicates whether the outer interval contains the multiplier values from the existing point completions.
\end{flushleft}
\end{table}

\section{CFS Moment Evidence and LP Duals}
\label{app:cfs_moment_evidence}

Table \ref{tab:home_bias_wedge} reports the sector-level home-share and held-out flow-fit comparisons behind Figure \ref{fig:home_bias_wedge}. The table uses the CFS interpretation in Section~\ref{sec:implementation}.

\begin{table}[htbp]
\centering
\caption{Home Bias and Held-Out Fit}
\label{tab:home_bias_wedge}
\normalsize
\begin{tabular}{lccccc}
\toprule
Supplier sector & CFS home & CI home & Gravity home & CI RMSE & Gravity RMSE \\
\midrule
Manufacturing & 0.371 & 0.026 & 0.284 & 0.0094 & 0.0055 \\
Mining & 0.567 & 0.032 & 0.323 & 0.0125 & 0.0098 \\
Transportation & 0.480 & 0.025 & 0.415 & 0.0097 & 0.0048 \\
Wholesale & 0.606 & 0.028 & 0.515 & 0.0165 & 0.0043 \\
\midrule
Mean & 0.506 & 0.028 & 0.384 & 0.0120 & 0.0061 \\
\bottomrule
\end{tabular}
\begin{flushleft}
\footnotesize Notes: CFS home is the within-state shipment share in the observed CFS origin-destination matrix for each strict shipment-covered sector. CI home is the home share under conditional independence. Gravity home is the home share under the structured spatial completion. RMSE columns report held-out normalized flow RMSE.
\end{flushleft}
\end{table}

Figure \ref{fig:distance_decay} compares observed CFS flows with the two main completions by distance bin. The proportional completion is diffuse because it uses only the marginals. Observed flows are much more concentrated in same-state and nearby bins. The structured completion moves the kernel toward this pattern.

\begin{figure}[htbp]
\centering
\includegraphics[width=0.95\textwidth]{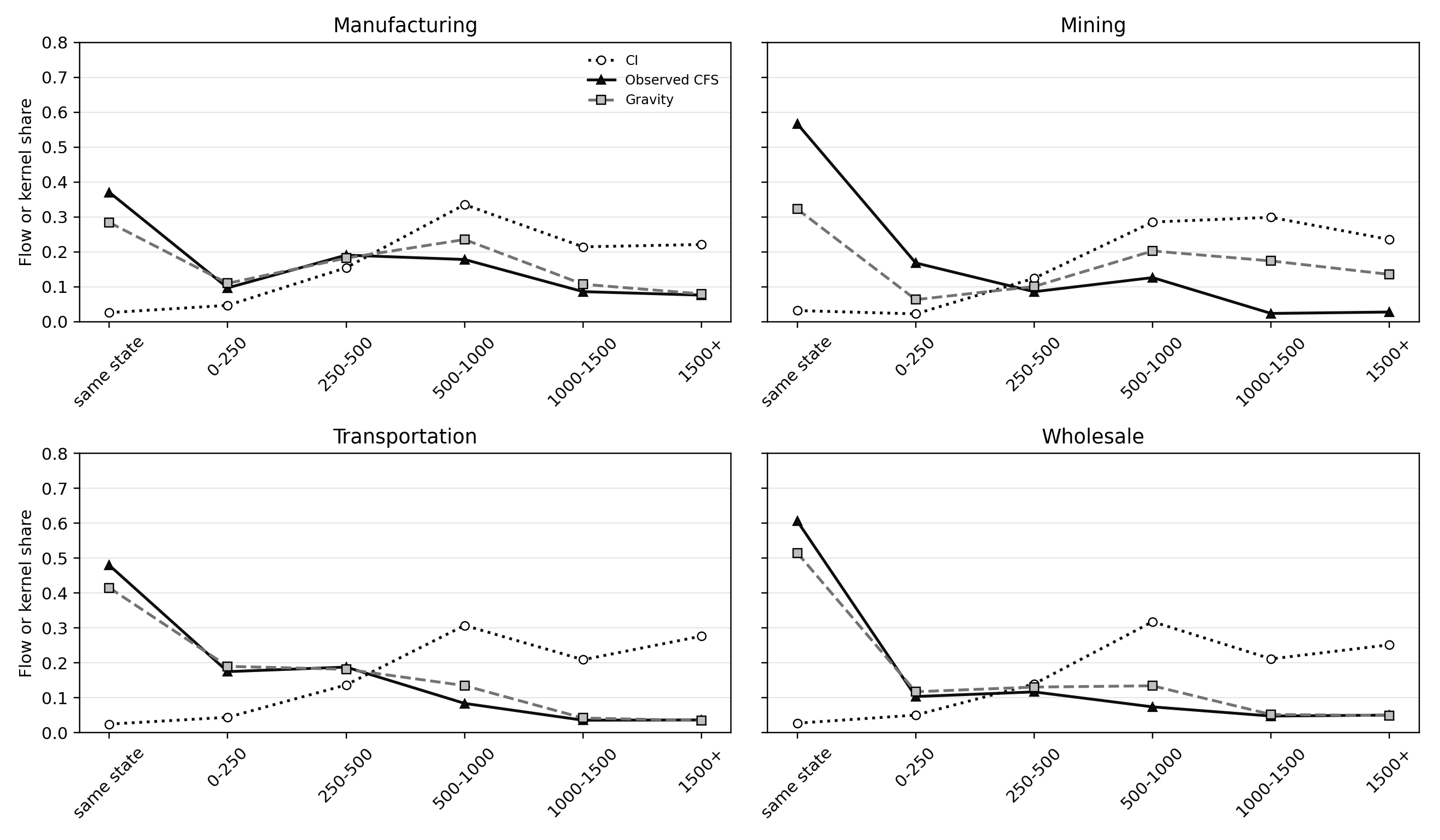}
\caption{Distance Decay in Observed Flows and Spatial Completions}
\label{fig:distance_decay}
\begin{flushleft}
\footnotesize Notes: The figure reports the share of flow or kernel mass in distance bins for the four strict shipment-covered sectors. The bins are same state, 0--250 miles, 250--500 miles, 500--1000 miles, 1000--1500 miles, and more than 1500 miles.
\end{flushleft}
\end{figure}

\begin{table}[htbp]
\centering
\caption{CFS Moment Feasibility Tolerances}
\label{tab:moment_tolerances}
\normalsize
\begin{tabular}{lcc}
\toprule
Supplier sector & Home-share $\tau$ & Distance-bin $\tau$ \\
\midrule
Manufacturing & $<10^{-6}$ & $<10^{-6}$ \\
Mining & 0.1414 & 0.1414 \\
Transportation & $<10^{-6}$ & $<10^{-6}$ \\
Wholesale & $<10^{-6}$ & $<10^{-6}$ \\
\bottomrule
\end{tabular}
\begin{flushleft}
\footnotesize Notes: The table reports the smallest uniform tolerance needed to make each CFS moment set feasible under the maintained wage-bill origin and intermediate-demand destination proxy marginals. Distance-bin tolerances apply jointly to the same-state and distance-bin shares. Mining requires the widest tolerance because the observed CFS home share is outside the range implied by the maintained marginals.
\end{flushleft}
\end{table}

\begin{table}[htbp]
\centering
\caption{Moment-Band Sensitivity for Gulf Exposure Bounds}
\label{tab:moment_band_bounds}
\normalsize
\setlength{\tabcolsep}{2.5pt}
\begin{tabular}{@{}>{\raggedright\arraybackslash}p{0.23\textwidth}>{\centering\arraybackslash}p{0.08\textwidth}>{\centering\arraybackslash}p{0.12\textwidth}>{\centering\arraybackslash}p{0.12\textwidth}>{\centering\arraybackslash}p{0.11\textwidth}>{\centering\arraybackslash}p{0.10\textwidth}>{\centering\arraybackslash}p{0.10\textwidth}@{}}
\toprule
Moment band & \shortstack{Formal\\CI} & \shortstack{Mean\\home\\$\tau$} & \shortstack{Mean\\distance\\$\tau$} & \shortstack{Median\\width} & \shortstack{P90\\width} & \shortstack{Possible\\top} \\
\midrule
Minimum feasibility band & No & 0.035 & 0.035 & 0.00309 & 0.00419 & 51 \\
CFS CV conservative band & No & 0.495 & 0.495 & 0.00312 & 0.00421 & 51 \\
Conservative robust band, 10 pp floor & No & 0.110 & 0.110 & 0.00309 & 0.00419 & 51 \\
\bottomrule
\end{tabular}
\begin{flushleft}
\footnotesize Notes: The table reports sharp one-step Gulf regional-shock exposure intervals under alternative bands around CFS home-share and distance-bin moments. Formal CI indicates whether the row is a formal confidence interval. It is ``No'' because the public CFS files do not provide the full covariance matrix for the constructed share vector. The bands describe uncertainty about auxiliary shipment moments, not sampling uncertainty for the full intermediate-input coupling.
\end{flushleft}
\end{table}

\begin{table}[htbp]
\centering
\caption{Largest LP Moment Shadows}
\label{tab:dual_shadow_prices}
\normalsize
\setlength{\tabcolsep}{2pt}
\begin{tabular}{@{}>{\raggedright\arraybackslash}p{0.08\textwidth}>{\raggedright\arraybackslash}p{0.17\textwidth}>{\raggedright\arraybackslash}p{0.08\textwidth}>{\centering\arraybackslash}p{0.12\textwidth}>{\centering\arraybackslash}p{0.10\textwidth}>{\raggedright\arraybackslash}p{0.24\textwidth}>{\centering\arraybackslash}p{0.10\textwidth}@{}}
\toprule
Target & Supplier & Bound & \shortstack{Sector LP\\endpoint} & \shortstack{Binding\\moments} & Top moment & Shadow \\
\midrule
CA & Transportation & Upper & 3.6e-04 & 6 & Dist. 1500+, upper & 0.0180 \\
TX & Mining & Upper & $<10^{-8}$ & 4 & Dist. same-state, lower & 0.0027 \\
LA & Mining & Lower & 6.2e-05 & 4 & Dist. same-state, lower & -0.0027 \\
CA & Mining & Upper & 5.3e-05 & 4 & Dist. same-state, lower & 0.0027 \\
TN & Mining & Upper & 4.8e-05 & 4 & Dist. same-state, lower & 0.0027 \\
IL & Mining & Upper & 5.3e-05 & 4 & Dist. same-state, lower & 0.0027 \\
\bottomrule
\end{tabular}
\begin{flushleft}
\footnotesize Notes: Rows are sector-specific LP dual calculations for listed target-sector-bound subproblems inside the state-level Gulf exposure bounds. They are not whole-state exposure endpoints or unique structural shadow prices. Binding moments counts nonzero moment multipliers. Top moment is the largest absolute moment multiplier in the subproblem. The maximum primal-dual gap in the full dual output is \DualMaxGap{}.
\end{flushleft}
\end{table}

The strongest binding bilateral cell is the \BindingTopCellSector{}
\BindingTopCellOD{} cell, which accounts for \BindingTopCellTargetPct{}
percent of sector CFS value and has a maximum absolute shadow price of
\BindingTopCellMaxShadow{} in these endpoint problems.

\begin{table}[htbp]
\centering
\caption{Binding Bilateral CFS Cell Restrictions}
\label{tab:binding_bilateral_cfs_cells}
\normalsize
\resizebox{\textwidth}{!}{\begin{tabular}{lccccc}
\toprule
Supplier sector & O-D cell & Target (\%) & Half-width & Feas. $\tau$ & Max shadow \\
\midrule
Manufacturing & LA-LA & 1.7 & 0.0054 & 0.0053 & 0.0691 \\
Wholesale & TX-TX & 10.8 & 0.0146 & 0.0145 & 0.0271 \\
Transportation & TX-TX & 7.7 & 0.0001 & 0 & 0.0160 \\
Transportation & FL-FL & 5.2 & 0.0001 & 0 & 0.0160 \\
Transportation & IL-IL & 2.4 & 0.0001 & 0 & 0.0160 \\
Transportation & PA-PA & 2.0 & 0.0001 & 0 & 0.0160 \\
Transportation & OH-OH & 1.7 & 0.0001 & 0 & 0.0160 \\
Transportation & GA-GA & 1.2 & 0.0001 & 0 & 0.0160 \\
Transportation & MO-MO & 0.9 & 0.0001 & 0 & 0.0160 \\
Wholesale & CA-CA & 9.4 & 0.0146 & 0.0145 & 0 \\
Transportation & CA-CA & 8.5 & 0.0001 & 0 & 0 \\
Manufacturing & CA-CA & 5.5 & 0.0054 & 0.0053 & 0 \\
\bottomrule
\end{tabular}}
\begin{flushleft}
\footnotesize Notes: The table reports selected bilateral CFS cell bands
that bind or have a nonzero dual shadow price in endpoint LPs for the Gulf
aggregate and the 20 states with the largest final upper endpoints. Each
listed cell binds in \BindingCellEndpointCount{} lower or upper endpoint LPs
across those targets. These cells are auxiliary shipment restrictions. They
are not observed intermediate-input buyer-seller links.
\end{flushleft}
\end{table}

\section{Robustness of Identified Exposure Bounds}
\label{app:identified_set_robustness}

\subsection{Mining moments}

Table \ref{tab:mining_robustness} asks whether the main state-level exposure bounds are driven by the mining moment, which requires the widest feasibility tolerance among the shipment-covered sectors. The comparison leaves the regional-incidence result essentially unchanged.

\begin{table}[htbp]
\centering
\caption{Mining Moment Robustness for Gulf Exposure Bounds}
\label{tab:mining_robustness}
\normalsize
\setlength{\tabcolsep}{2.5pt}
\begin{tabular}{@{}>{\raggedright\arraybackslash}p{0.28\textwidth}>{\centering\arraybackslash}p{0.11\textwidth}>{\centering\arraybackslash}p{0.12\textwidth}>{\centering\arraybackslash}p{0.10\textwidth}>{\centering\arraybackslash}p{0.10\textwidth}>{\centering\arraybackslash}p{0.10\textwidth}@{}}
\toprule
Maintained CFS moments & \shortstack{Agg.\\width} & \shortstack{Median\\width} & \shortstack{P90\\width} & \shortstack{Robust\\top} & \shortstack{Possible\\top} \\
\midrule
CFS moments for all shipment sectors & 0.00000 & 0.00309 & 0.00419 & 0 & 51 \\
Mining CFS moments excluded & 0.00000 & 0.00312 & 0.00421 & 0 & 51 \\
\bottomrule
\end{tabular}
\begin{flushleft}
\footnotesize Notes: The table reports sharp one-step exposure bounds for the Gulf regional shock under the final CFS moment set and under a variant that excludes mining CFS moments. The comparison asks whether the wide mining feasibility tolerance drives the state-level exposure intervals.
\end{flushleft}
\end{table}

\clearpage

\subsection{Local sectors and retail}

For local sectors, we impose support restrictions rather than estimating a full bilateral coupling from shipment data. The appendix reports the feasible range of within-state home shares. The baseline support includes same-state pairs, county-adjacent states, and the minimum additional nearest-state radius needed for the origin and destination marginals to be feasible. Figure \ref{fig:local_bounds} reports the resulting intervals. The mean interval width is \LocalMeanIntervalWidth.

\begin{figure}[htbp]
\centering
\includegraphics[width=0.85\textwidth]{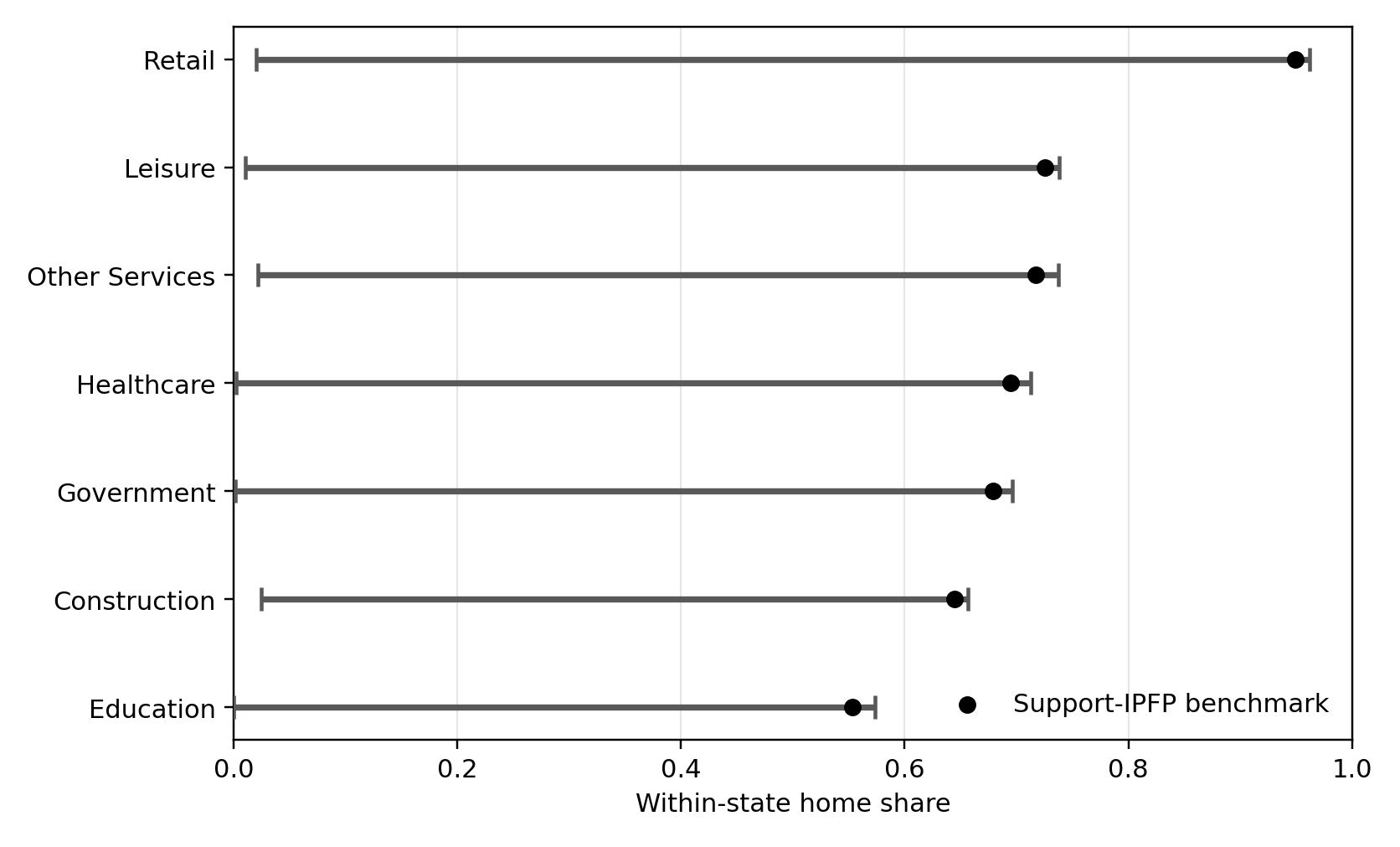}
\caption{Identified Sets for Local-Sector Home Shares}
\label{fig:local_bounds}
\begin{flushleft}
\footnotesize Notes: Horizontal lines are feasible intervals for within-state home shares. Dots are distance-weighted support-IPFP completions.
\end{flushleft}
\end{figure}

\clearpage

Table \ref{tab:retail_robustness} shows how results change if retail is treated as shipment-informed. Treating retail as shipment-informed improves held-out shipment fit, but it changes the target toward goods movement rather than local retail service provision. The main specification therefore keeps retail in the local support-restricted group.

\begin{table}[htbp]
\centering
\caption{Retail Classification Robustness}
\label{tab:retail_robustness}
\normalsize
\begin{tabular}{lccc}
\toprule
Retail treatment & Home share & Average distance & Held-out RMSE \\
\midrule
Observed CFS & 0.309 & 696.8 &  \\
CI & 0.027 & 1141.7 & 0.0063 \\
Main local-support & 0.950 & 16.7 & 0.0122 \\
Appendix shipment-informed & 0.231 & 736.3 & 0.0053 \\
\bottomrule
\end{tabular}
\begin{flushleft}
\footnotesize Notes: Observed CFS summarizes retail-related shipments. CI is the proportional completion. Main local-support is the support-restricted retail treatment used in the baseline specification. Appendix shipment-informed uses the retail CFS flows as auxiliary goods-movement evidence. The comparison concerns sector classification and does not imply that retail shipments observe local retail service provision.
\end{flushleft}
\end{table}

\subsection{Margins and state rankings}

Table \ref{tab:alternative_margin_identified_sets} reports the identified-set calculation under wage-bill, employment, mixed, and model-output marginals. It supports the main-text result that the exact-margin aggregate cancellation depends on the maintained aggregate target, while state-level incidence remains weakly identified across margin choices.

\begin{table}[htbp]
\centering
\caption{Alternative Marginals and Gulf Exposure Bounds}
\label{tab:alternative_margin_identified_sets}
\normalsize
\resizebox{\textwidth}{!}{\begin{tabular}{lccccc}
\toprule
Margin construction & Agg. width & Median width & P90 width & Possible top & Max $\tau$ \\
\midrule
Wage-bill origin, wage-bill destination & 0.00000 & 0.00309 & 0.00419 & 51 & 0.141 \\
Employment origin, employment destination & 0.00383 & 0.00349 & 0.00556 & 51 & 0.107 \\
Wage-bill origin, employment destination & 0.00320 & 0.00326 & 0.00467 & 51 & 0.160 \\
Employment origin, wage-bill destination & 0.00000 & 0.00368 & 0.00504 & 51 & 0.101 \\
Model-output origin, model-output destination & 0.02119 & 0.00346 & 0.00805 & 51 & 0.154 \\
\bottomrule
\end{tabular}}
\begin{flushleft}
\footnotesize Notes: The table repeats the sharp Gulf exposure bounds under alternative origin and destination marginals. Wage-bill origin and wage-bill destination is the baseline. The model-output row is a model-implied robustness margin, not the baseline. State-level top-decile exposure remains weakly identified in every row.
\end{flushleft}
\end{table}

\clearpage

\subsection{Manufacturing split robustness}

Table~\ref{tab:manufacturing_split_feasibility} reports the coverage measures for
the five-part manufacturing split. Each proposed manufacturing group has
positive QCEW wage-bill and employment coverage in all 51 states. Detailed
CFS manufacturing rows mapped to the five proposed groups account for
\ManufacturingSplitDetailedCfsCoveragePct{} percent of aggregate
manufacturing CFS value. Each group has at least
\ManufacturingSplitMinCfsPairs{} positive state-pair cells, and the largest
top-15-cell share is \ManufacturingSplitMaxTopFifteenCfsPct{} percent.

The twenty-sector robustness specification reconstructs the national IO block,
QCEW wage-bill and employment margins, CFS moments, selected bilateral CFS
cell bands, and the Gulf exposure LPs. Table
\ref{tab:manufacturing_split_identified_sets} reports the resulting
identified-set comparison. Under exact wage-bill margins and bilateral CFS
bands, the twenty-sector median state width is
\ManufacturingSplitMedianWidth{} and the p90 width is
\ManufacturingSplitPninetyWidth{}. Sharp pairwise dominance determines
\ManufacturingSplitPairwiseDeterminedPct{} percent of state-pair rankings.
Under banded margins, aggregate width is
\ManufacturingSplitBandedAggregateWidth{} and median state width is
\ManufacturingSplitBandedMedianWidth{}.

The split is feasible, but it does not justify replacing the sixteen-sector
baseline. Finer manufacturing detail does not necessarily sharpen local
incidence because the remaining uncertainty is not concentrated only in
manufacturing. In the twenty-sector specification, split manufacturing subsectors
account for \ManufacturingSplitManufacturingWidthSharePct{} percent of total
state interval width. The largest source is \ManufacturingSplitTopWidthGroup{},
which accounts for \ManufacturingSplitTopWidthGroupPct{} percent. The result
is the same as in the main specification: additional goods-sector
detail helps only if it addresses the sectors and moments that bind the local
exposure intervals.

\begin{table}[htbp]
\centering
\caption{Manufacturing Split Feasibility}
\label{tab:manufacturing_split_feasibility}
\normalsize
\resizebox{\textwidth}{!}{\begin{tabular}{lcccccccc}
\toprule
Manufacturing group & BEA output (\%) & BEA intermediate (\%) & QCEW wage bill (\%) & Detailed CFS (\%) & CFS pairs & Top 15 CFS cells (\%) & Activity coverage & CFS coverage \\
\midrule
Food, textile, wood, paper, printing & 22.9 & 25.8 & 19.7 & 26.1 & 2051 & 29.3 & Yes & Yes \\
Petroleum, chemicals, plastics & 27.8 & 28.1 & 15.4 & 29.0 & 1790 & 44.7 & Yes & Yes \\
Nonmetallic minerals, metals, machinery & 19.1 & 19.1 & 24.6 & 19.1 & 1936 & 27.1 & Yes & Yes \\
Computer, electrical, transport equipment & 26.0 & 23.7 & 33.6 & 21.9 & 1702 & 31.0 & Yes & Yes \\
Furniture and miscellaneous manufacturing & 4.1 & 3.3 & 6.7 & 3.9 & 1577 & 24.4 & Yes & Yes \\
\bottomrule
\end{tabular}}
\begin{flushleft}
\footnotesize Notes: Activity coverage requires positive QCEW wage-bill and
employment activity in all 51 states. Detailed CFS is the group's share of
detailed manufacturing CFS value after mapping three-digit manufacturing
records to the proposed groups. CFS pairs is the number of positive state-pair cells.
Top 15 CFS cells reports the shipment-value share accounted for by the 15
largest positive state-pair cells in the group. CFS coverage requires at least
15 positive state-pair cells.
\end{flushleft}
\end{table}

\begin{table}[htbp]
\centering
\caption{Manufacturing Split Identified-Set Robustness}
\label{tab:manufacturing_split_identified_sets}
\normalsize
\resizebox{\textwidth}{!}{\begin{tabular}{lcccccc}
\toprule
Admissible set & Agg. width & Median width & P90 width & Pairwise det. (\%) & Robust top & Possible top \\
\midrule
Margins only & 0.00000 & 0.00462 & 0.00676 &  & 0 & 51 \\
Local support & 0.00000 & 0.00433 & 0.00618 &  & 0 & 51 \\
CFS home band & 0.00000 & 0.00431 & 0.00590 &  & 0 & 51 \\
CFS distance-bin band & 0.00000 & 0.00431 & 0.00590 &  & 0 & 51 \\
CFS bilateral cell bands & 0.00000 & 0.00431 & 0.00590 & 0.0 & 0 & 51 \\
Banded margins + CFS bilateral & 0.00227 & 0.00536 & 0.00789 &  & 0 & 51 \\
\bottomrule
\end{tabular}}
\begin{flushleft}
\footnotesize Notes: The table repeats the Gulf exposure bound calculation
after splitting aggregate manufacturing into five manufacturing subsectors.
Pairwise dominance is recomputed only for the exact-margin bilateral-CFS
row. The sixteen-sector specification remains the main specification because
the split does not sharpen local-incidence bounds and because the paper's
point-completion comparisons use the sixteen-sector production block.
\end{flushleft}
\end{table}

\end{document}